\documentclass[5p,times]{elsarticle}

\usepackage{algorithm}
\usepackage[noend]{algpseudocode}
\usepackage{overpic} 
\usepackage{graphicx} 
\usepackage{amsmath}
\usepackage{array}
\usepackage{float}
\usepackage{amssymb}
\usepackage{comment}
\usepackage{xcolor}
\usepackage[title]{appendix}
\usepackage{amsthm}
\usepackage{amsfonts}
\usepackage{algorithmicx}

\newtheorem{proposition}{Proposition}

\begin{document}

\begin{frontmatter}

\title{Efficient Resource Allocation in 5G Massive MIMO-NOMA Networks: Comparative Analysis of SINR-Aware Power Allocation and Spatial Correlation-Based Clustering}

\author[unilim]{Samar Chebbi\corref{cor1}}
\ead{samar.chebbi@unilim.fr}

\author[uca]{Oussama Habachi}
\ead{oussama.habachi@uca.fr}

\author[unilim]{Jean-Pierre Cances}
\ead{jean-pierre.cances@xlim.fr}

\author[unilim]{Vahid Meghdadi}
\ead{vahid.meghdadi@unilim.fr}

\author[teluq]{Essaid Sabir\fnref{also}}
\ead{essaid.sabir@teluq.ca}

\cortext[cor1]{Corresponding author.}

\address[unilim]{XLIM, University of Limoges, 87060 Limoges, France.}
\address[uca]{LIMOS, University of Clermont Auvergne, 63178 Aubiere, France.}
\address[teluq]{TÉLUQ, University of Quebec Montreal, Montreal, Canada.}

\begin{abstract}
With the evolution of 5G networks, optimizing resource allocation has become crucial to meeting the increasing demand for massive connectivity and high throughput. Combining Non-Orthogonal Multiple Access (NOMA) and massive Multi-Input Multi-Output (MIMO) enhances spectral efficiency, power efficiency, and device connectivity. However, deploying MIMO-NOMA in dense networks poses challenges in managing interference and optimizing power allocation while ensuring that the Signal-to-Interference-plus-Noise Ratio (SINR) meets required thresholds.
Unlike previous studies that analyze user clustering and power allocation techniques under simplified assumptions, this work provides a comparative evaluation of multiple clustering and allocation strategies under identical spatially correlated network conditions. We focus on maximizing the number of served users under a given Quality of Service (QoS) constraint rather than the conventional sum-rate maximization approach. Additionally, we consider spatial correlation in user grouping, a factor often overlooked despite its importance in mitigating intra-cluster interference.
We evaluate clustering algorithms, including user pairing, random clustering, Correlation Iterative Clustering Algorithm (CIA), K-means++-based User Clustering (KUC), and Grey Wolf Optimizer-based clustering (GWO), in a downlink spatially correlated MIMO-NOMA environment. Numerical results demonstrate that the GWO-based clustering algorithm achieves superior energy efficiency while maintaining scalability, whereas CIA effectively maximizes the number of served users. These findings provide valuable insights for designing MIMO-NOMA systems that optimize resource allocation in next-generation wireless networks.
\end{abstract}


\end{frontmatter}

\section{Introduction}
The advent of 5G marks a paradigm shift in wireless communications, promising faster data speeds, improved reliability, and massive connectivity for billions of devices. However, the growing demand for high throughput and massive access presents significant challenges to existing network architectures. Non-Orthogonal Multiple Access (NOMA) and massive Multi-Input Multi-Output (MIMO) are key technologies that significantly improve spectral efficiency, power utilization, and network capacity. However, deploying MIMO-NOMA in dense wireless networks requires effective resource allocation strategies such as user clustering and power allocation, to manage interference and maintain optimal performance.
Unlike conventional approaches that prioritize sum-rate maximization, our focus is on maximizing the number of served users while ensuring a given Quality of Service (QoS) constraint. Furthermore, despite its crucial role in mitigating intra-cluster interference, spatial correlation among users has been largely overlooked in previous user grouping strategies.
While extensive research has explored various user clustering and power allocation methods for MIMO-NOMA systems, most studies fail to provide a unified comparative analysis of clustering algorithms under consistent spatially correlated network conditions. 
This study bridges these gaps by systematically evaluating multiple user clustering algorithms—including user pairing, random clustering, the Correlation Iterative Clustering Algorithm (CIA), K-means++-based User Clustering (KUC), and Grey Wolf Optimizer-based clustering (GWO)—within a common downlink MIMO-NOMA framework that explicitly accounts for spatial correlation. 
The key contributions of this paper are summarized as follows:
\begin{itemize}
\item We conducted an extensive comparative evaluation of multiple user clustering algorithms for MIMO-NOMA, analyzing their trade-offs in interference management, power efficiency, and scalability while ensuring compliance with SINR thresholds.
\item We assessed these algorithms based on critical performance metrics, including the number of users served, power efficiency, and scalability, providing insights into their effectiveness in spatially correlated environments.
\item To enhance network reliability and mitigate inter-cluster interference, we employ advanced beamforming techniques, such as Zero-Forcing, ensuring robust communication in dense user environments.
\end{itemize}

The remainder of this paper is organized as follows. Section II reviews the background and related work. Section III details the proposed power allocation strategy, while Section IV discusses user clustering and power control in MIMO-NOMA systems. Section V presents a computational complexity analysis, followed by the performance evaluation in Section VI. Section VII discusses challenges and limitations, while Section VIII outlines future research directions. Finally, Section IX concludes the paper with a summary of our findings and their implications for next-generation wireless networks.


\section{Background and Related Work}

Despite the evolution of 5G technology, it is still expected to meet the growing demand for Internet connectivity and bandwidth challenges. Current research directions investigate the combination of innovative technologies such as NOMA and massive MIMO, which are crucial to increasing spectral efficiency. For example, Benjebbour et al. and Alberto et al. outlined, in \cite{mimo_noma_2} and \cite{mimo_noma_1}, the potential of combining NOMA and MIMO, highlighting their crucial roles in addressing exponential growth in traffic and the number of connected devices. 

Various user clustering approaches have been proposed to improve the performance of MIMO-NOMA systems. As a baseline approach, we consider user pairing, which is widely considered to enhance system performance with low computational complexity. One notable approach is Near-Far user pairing, which pairs users based on their channel gains relative to the base station. This strategy involves pairing a user with a high channel gain (close to the base station) with a user with a low channel gain (far from the base station), leveraging the channel gain disparities to maximize spectral efficiency.
Ding et al. \cite{up} explored the theoretical foundations of Near-Far user pairing in the context of downlink NOMA transmissions and proved that this strategy significantly improves spectral efficiency by reducing interference and improving the signal-to-interference-plus-noise ratio (SINR) for far users. Based on this, Zhang et al. \cite{up_sic} proposed a user pairing based on channel gain. This method improves system performance but suffers from high computational complexity, making it challenging for practical implementation. On the other hand, El-ghorab et al. \cite{ee_up} focused on energy-efficient user pairing for downlink NOMA in massive MIMO networks. Their proposed algorithms reduce the overall power consumption while maintaining high levels of fairness among users.

Recent contributions highlighted the importance of spatial correlation in improving the effectiveness of beamforming and interference mitigation by leveraging the benefits of channel correlation and gain disparities.
In fact, they have integrated spatial correlation into user clustering strategies to reduce inter-user interference and improve the efficiency of NOMA-aided MIMO systems. For instance, Kiani et al. \cite{sc_1} proposed a learning-based user clustering algorithm that focuses on minimizing power consumption while ensuring the transmission quality required by the massive MIMO system. Integrating spatial correlation in their method has significantly enhanced the network's ability to manage interference, improving overall system capacity and energy efficiency. Meanwhile  \cite{adaptive_noma} introduced an adaptive user clustering approach for uplink NOMA-based IoT networks, where users are grouped according to spatial correlation to ensure that users within the same cluster have less correlated channel conditions to improve NOMA decoding efficiency. This approach improves the signal separation process at the receiver and reduces interference among clustered users, making it particularly suitable for uplink IoT scenarios.
Moreover, Wang et al. \cite{sc_2} explored the outage probability in a Simultaneously Transmitting And Reflecting Reconfigurable Intelligent Surface (STAR-RIS) assisted NOMA network operating over spatially correlated channels. Their research underscores the importance of spatial correlation in designing and optimizing NOMA systems. Considering the spatial characteristics of the channels, they achieved a significant improvement in the outage probability.
On the other hand, Shafie et al. \cite{sc_3} investigated power allocation strategies for a High-Altitude Platform Station (HAPS)-enabled MIMO-NOMA system with spatially correlated channels. Their work emphasizes the importance of utilizing customized power allocation techniques that are based on spatial correlation, which is essential to maintain system performance in such complex environments.
Furthermore, in \cite{Dhakal2020LineOfSight}, Dhakal et al. provided insights on how line-of-sight components and spatial correlation impact NOMA performance, opening new avenues for research on NOMA beamforming and clustering. They investigated how semi-orthogonal users and highly correlated users within a beam can maximize spectral efficiency and reduce interference of NOMA systems. Their comparative analysis of Rician and Rayleigh channels provides valuable perspectives on the intricacies of clustering in correlated channels. Despite their contributions, the study acknowledges the challenges posed by static clustering criteria and the need for methods that could adapt to dynamic user distributions and mobility patterns.

Machine learning techniques, especially K-means, are crucial to improving resource allocation strategies within next-generation network systems. For example, the work of Cabrera et al. \cite{cabrera2018enhanced} emphasizes the use of the K-means algorithm to optimize the sum rate of the MIMO-NOMA system while ensuring fairness among users. Chu et al. \cite{chu2023sum} focused on improving the efficiency of downlink communication in high-mobility scenarios using orthogonal time frequency space (OTFS) modulation with NOMA. The K-means algorithm partitioned the base station coverage area, ensuring that users are close to the cluster center, thus reducing mutual interference through strategic power allocation. 
Shukla et al. \cite{shukla2022kmeans++} introduced the K-means++ algorithm to address the computational challenge by providing a more refined initialization process, which has been shown to converge more rapidly to optimal solutions. On the other hand, Authors in \cite{rdn_interference} proposed a DRL-based framework that jointly manages interference and dynamically adjusts power allocation to improve spectral efficiency. By leveraging deep Q-learning and policy gradient techniques, the approach continuously adapts to changing network conditions. The study demonstrated significant improvements in spectral efficiency and interference mitigation, particularly in high-density environments, making it distinct from traditional heuristic-based methods.

While various studies have explored different user clustering and power allocation strategies, a major gap remains in providing a comprehensive comparative analysis of these methods under diverse network conditions, particularly in spatially correlated environments. Our research specifically aims to provide a comparative evaluation of our previously developed user clustering and power allocation techniques within a common network framework, specifically under the same spatially correlated network conditions.

\section{System Model and Preliminaries}

We consider a downlink MIMO system based on the NOMA architecture, where the Base Station (BS), equipped with $M$ Radio Frequency (RF) chains, serves $N$ single-antenna receivers.  To maximize multiplexing gains, the number of user clusters is set equal to the number of RF chains, ensuring efficient resource allocation and interference management.  
To accurately model the MIMO-NOMA wireless channel, we adopt an advanced millimeter-wave (mmWave) representation based on the extended Saleh-Valenzuela model, which effectively characterizes the propagation dynamics of mmWave communications. A defining feature of the extended Saleh-Valenzuela model is its cluster-based propagation structure, where mmWave signals travel through multiple clusters of rays. Each cluster contains several rays with slight variations in their angles of arrival and departure, introducing spatial correlation effects that significantly impact interference and beamforming performance. The angular spread of these rays is modeled using a Laplacian distribution to accurately reflect real-world mmWave propagation characteristics. 
By incorporating the transmit and receive antenna gains along with normalized array response vectors, this model effectively captures the interplay between beamforming, interference alignment, and power allocation. As a result, it provides a robust framework for evaluating the performance of MIMO-NOMA systems in dense wireless environments.  

For user $i \in \{1, \ldots, N\}$ and cluster $k \in \{1, \ldots, M\}$, the channel matrix $\mathbf{H}_{i,k} \in \mathbb{C}^{N \times M}$ is expressed as:

\begin{equation}
    \mathbf{H}_{i,k} = \gamma \sum_{i=1}^{M} \sum_{k=1}^{N} \alpha_{i,k} \mathbf{G}_{r,t}(\theta^{r,t}_{i,k}, \phi^{r,t}_{i,k}) \mathbf{a}_{r,t}(\theta^{r,t}_{i,k}, \phi^{r,t}_{i,k})
\end{equation}

where $\gamma$ is a normalization factor, $\alpha_{i,k}$ is the complex gain of the $k^{th}$ ray in the $i^{th} $scattering cluster, and $\phi^{r}$ and $\phi^{t}$ are, respectively, the azimuth and elevation angles of arrival and departure. The function $\mathbf{G}_{r,t}(\phi^{r,t}_{i,k}, \theta^{r,t}_{i,k})$ represents the combined gain of the transmit and receive antennas at the corresponding angles, and $\mathbf{a}_{r,t}(\theta^{r,t}_{i,k}, \phi^{r,t}_{i,k})$ denotes the normalized response vectors of the receiver and transmitter at these angles.

\subsection{Zero-Forcing Beamforming}  

In the considered massive MIMO-NOMA downlink system, each cluster is served with a single transmission beam. Massive MIMO enables simultaneous transmission of multiple data streams through different antennas, while NOMA allows multiple users within the same cluster to share spectral resources by assigning different power levels. Users within each cluster are selected based on high channel correlation, ensuring that their channel responses exhibit strong similarity due to common propagation effects such as scattering, reflection, and fading. Among them, the user with the highest correlation to the rest of the cluster members is designated as the representative user, responsible for generating the Zero-Forcing (ZF) precoding matrix. The ZF beamforming matrix $\mathbf{W}_{i,k}$ is designed to suppress inter-cluster interference by directing the transmission beam precisely towards the representative user of cluster $k$. It is computed as the pseudo-inverse of the channel matrix $\mathbf{H}_{i,k}$, ensuring that the transmitted signals are orthogonal to the interference space:

\begin{equation}
\mathbf{W}_{i,k} = \mathbf{H}_{i,k}^{\dagger} = \left( \mathbf{H}_{i,k}^{H} \mathbf{H}_{i,k} \right)^{-1} \mathbf{H}_{i,k}^{H}
\end{equation}

where $\mathbf{H}_{i,k}^{\dagger}$ represents the Moore-Penrose pseudo-inverse, $\mathbf{H}_{i,k}^{H}$ is the Hermitian transpose, and $\mathbf{H}_{i,k}$ is the channel matrix corresponding to the representative user. While ZF effectively eliminates interference for the representative user, the spatial correlation among cluster members naturally reduces intra-cluster interference for other users, as their signal paths exhibit similar propagation characteristics. This spatial correlation improves the efficiency of Successive Interference Cancellation (SIC) by ensuring more uniform interference levels within each cluster, thereby enhancing decoding performance. Additionally, cluster formation is designed to minimize collinearity between different groups, further reducing inter-cluster interference.  

To prevent excessive power consumption and maintain stable beamforming performance, the ZF precoding matrix is normalized using its Frobenius norm:

\begin{equation}
\mathbf{W}_{i,k}^{\text{norm}} = \frac{\mathbf{W}_{i,k}}{\| \mathbf{W}_{i,k} \|_F}
\end{equation}

where

\begin{equation}
\| \mathbf{W}_{i,k} \|_F = \sqrt{\sum_{i=1}^{M} \sum_{j=1}^{N} |\mathbf{W}(i,j)|^2}
\end{equation}

ensuring that the beamforming power remains within operational constraints. By leveraging ZF beamforming alongside correlation-based user clustering, the proposed system maximizes the number of served users while effectively mitigating inter-cluster interference, leading to improved overall performance in dense MIMO-NOMA environments.

\subsection{ Spatial correlation}
\label{spatial correlation}
In a mMIMO-NOMA network, spatial correlation is crucial for user clustering since it directly affects how effectively users can be grouped to minimize interference and optimize resource allocation. 
Spatial correlation refers to the degree of similarity in signal characteristics experienced by users who experience similar environmental conditions. High spatial correlation occurs when the channel responses of different users are very similar, typically because the users are close to each other or share similar propagation environments, resulting in similar fading and interference patterns. However, a low spatial correlation occurs when users are spread out in different locations or experience different propagation paths, resulting in different channel behavior. This orthogonality guarantees that the signals are less likely to interfere with each other, even when they share the same frequency band or time slot. Consequently, separating and managing signals from different users becomes much easier, thereby reducing interference and improving overall system performance.
The spatial correlation between the channels of users $i$ and $j$ can be expressed as:
\begin{equation}
\label{eq:coll}
C_{ij} = \frac{E[\mathbf{h}_i^H \mathbf{h}_j]}{\sqrt{E[|\mathbf{h}_i|^2] E[|\mathbf{h}_j|^2]}},
\end{equation}
where $\mathbf{h}_i$ and $\mathbf{h}_j$ are the channel vectors of users $i$ and $j$, respectively. When the channels are highly correlated (i.e., \( C_{ij} \) is close to 1),  the signals experience similar environmental effects, such as scattering, reflection, and fading. On the other hand, a low correlation (that is, \( C_{ij} \) is close to 0) indicates that the channels are more orthogonal, meaning that the signals are less likely to interfere with each other, which is advantageous for system performance.

\section{Power Allocation Algorithm}

In this section, we formulate the power allocation problem with the objective of maximizing the number of users whose SINR exceeds the required threshold \( \gamma_{\text{th}} \), ensuring reliable communication under total power constraints. Unlike conventional approaches that primarily focus on maximizing energy efficiency or overall system throughput, our method prioritizes serving the maximum number of users while maintaining their Quality of Service (QoS) requirements.

To achieve this, power is allocated efficiently among users within each cluster, avoiding excessive distribution to users with poor channel conditions, which could lead to overall performance degradation. Instead, resources are allocated optimally to users who can sustain reliable communication, ensuring an improved balance between spectral efficiency, fairness, and system capacity. 
Based on these considerations, the power allocation problem is formulated as follows:

\begin{equation}
\label{eq: power_system_cluster}
\mathbf{P}^{*} =
\left\{
    \begin{aligned}
        & \underset{\mathbf{P}}{\arg \max} \sum_{k=1}^{M} \sum_{j \in \zeta_k} \mathcal{U}(P_{j,k}) \\
        & \text{subject to:} \\
        & \quad \quad \sum_{k=1}^{M} \sum_{j \in \zeta_k} P_{j,k} \leq P_{\text{max}}, \\
        & \quad \quad \gamma_j \geq \gamma_{\text{th}}, \quad \forall j \in \zeta_k, \, \forall k \in \{1, 2, \dots, M\}.
    \end{aligned}
\right.
\end{equation}
where \( \mathcal{U}(P_{j,k}) \) is the utility function:
\begin{equation}
\mathcal{U}(P_{j,k}) =
\begin{cases} 
1, & \text{if } \gamma_j \geq \gamma_{\text{th}} \\
0, & \text{if } \gamma_j < \gamma_{\text{th}} 
\end{cases}
\end{equation}

The SINR of user \( j \), denoted as \( \gamma_j \), is expressed as:

\begin{equation}
\small
\gamma_j = \frac{C_{j} P_{j,k}|\mathbf{H}_{j,k} \mathbf{W}_{j,k}^{\text{H}}|^{2}}{\sigma^{2} + C_{j} |\mathbf{H}_{i,k} \mathbf{W}_{j,k}^{\text{H}}|^{2} \sum_{i \in \zeta_k, i < j} P_{i,k} + (1-C_{j})|\mathbf{H}_{j,k} \mathbf{W}_{j,k}^{\text{H}}|^{2} \sum_{\substack{l=i+1 \\ m \not \in \zeta_k}}^{N} P_{l,m}} 
\end{equation}

In this expression \( C_{j} \) represents the collinearity coefficient between the user \( j \) in the cluster \( k \) and the representative user of the cluster, \( p_{j,k} \) is the power coefficient allocated to the user \( j \) within the cluster \( k \), \( \mathbf{H}_{j,k} \) is the channel matrix and \( \mathbf{W}_{j,k} \) is the beamforming vector related to the user \( j \), while \( \sigma^{2} \) denotes the noise variance.

\begin{proposition}
In order to meet the power requirements of all clusters, the total allocated power should be bounded by $P_{\text{Max}}$ and $P_{\text{min}}$, 
where \( P_{\text{min}} \) represents the minimum power required to meet the SINR thresholds and QoS requirements in all clusters. \( P_{\text{min}} \) is illustrated in Equation~(\ref{eq: closed_form_pmin}) if user have the same QoS requirements and in Equation~(\ref{eq: closed_form_pmin_variable}) otherwise.
\end{proposition}
\begin{proof}
see Appendix~\ref{A1}
\end{proof}

\begin{figure*}[!htbp]
\begin{equation}
\label{eq: closed_form_pmin}
    \begin{aligned}    
    P_{\text{min}}=\frac
    {\sum_{k=1}^{M} \left(\frac{\sum_{i=1}^{n_{k}} \frac{ \gamma_{\text{th}} \sigma^{2}}{\mathbf{C}_{i} |\mathbf{H}_{i,k} \mathbf{W}_{i,k}^{\text{H}}|^{2}} + \gamma_{\text{th}}^{2} \sigma^{2} \sum_{i=1}^{n_k} \sum_{j=1}^{i-1} \frac{(1+\gamma_{\text{th}})^{i-j-1}}{\mathbf{C}_{j} |\mathbf{H}_{j} \mathbf{W}_{j}^{\text{H}}|^{2}}}{1+ \sum_{i=1}^{n_{k}} \gamma_{\text{th}} \frac{(1-\mathbf{C}_{i})}{\mathbf{C}_{i}} \sum_{m \not \in \zeta_k} \sum_{\substack{l=n_k+1 \\ m \not \in \zeta_k}}^{N} P_{l,m}  + \gamma_{\text{th}}^{2} \sigma^{2} \sum_{i=1}^{n_i} \sum_{j=1}^{i-1} (1+\gamma_{\text{th}})^{i-j-1}\frac{(1-\mathbf{C}_{j})}{\mathbf{C}_{j}} \sum_{m \not \in \zeta_k} \sum_{\substack{l=n_k+1 \\ m \not \in \zeta_k}}^{N} P_{l,m} }\right)} 
    { 1- \sum_{k=1}^{M} \left(\frac{\sum_{i=1}^{n_{k}}(\gamma_{\text{th}} \frac{(1-\mathbf{C}_{i})}{\mathbf{C}_{i}} \sum_{m \not \in \zeta_k} \sum_{\substack{l=n_k+1 \\ m \not \in \zeta_k}}^{N} P_{l,m}  + \gamma_{\text{th}}^{2} \sigma^{2} \sum_{i=1}^{n_k} \sum_{j=1}^{i-1} (1+\gamma_{\text{th}})^{i-j-1}\frac{(1-\mathbf{C}_{j})}{\mathbf{C}_{j}} \sum_{m \not \in \zeta_k} \sum_{\substack{l=n_k+1 \\ m \not \in \zeta_k}}^{N} P_{l,m}}{1+\sum_{i=1}^{n_{k}}\gamma_{\text{th}} \frac{(1-\mathbf{C}_{i})}{\mathbf{C}_{i}} \sum_{m \not \in \zeta_k} \sum_{\substack{l=n_k+1 \\ m \not \in \zeta_k}}^{N} P_{l,m} + \gamma_{\text{th}}^{2} \sigma^{2} \sum_{i=1}^{n_k} \sum_{j=1}^{i-1} (1+\gamma_{\text{th}})^{i-j-1}\frac{(1-\mathbf{C}_{j})}{\mathbf{C}_{j}}\sum_{m \not \in \zeta_k} \sum_{\substack{l=n_k+1 \\ m \not \in \zeta_k}}^{N} P_{l,m}}    
    \right)}  
    \end{aligned}
\end{equation}
\end{figure*}

\begin{figure*}[!htbp]
\begin{equation}
\label{eq: closed_form_pmin_variable}
\resizebox{\textwidth}{!}{$ 
\begin{aligned}
P_{\text{min}} = \frac{
\sum_{k=1}^{M} \left( \frac{
\sum_{i=1}^{n_k} \left( \frac{\gamma_{\text{th},i} \sigma^2}{C_i |\mathbf{H}_{i,k} \mathbf{W}_{i,k}^{\text{H}}|^2} \right) 
+ \sum_{i=1}^{n_k} \sum_{j=1}^{i-1} \gamma_{\text{th},i} \gamma_{\text{th},j} \frac{\sigma^2}{C_j |\mathbf{H}_{j,k} \mathbf{W}_{j,k}^{\text{H}}|^2} + \sum_{i=1}^{n_k} \sum_{j=1}^{i-1} \sum_{l=1}^{j-1} \gamma_{\text{th},i} \gamma_{\text{th},j} \gamma_{\text{th},l} \frac{\sigma^2}{\mathbf{C}_l |\mathbf{H}_{l,k} \mathbf{W}_{l,k}^{\text{H}}|^2}
}{1 + 
\sum_{i=1}^{n_k} \gamma_{\text{th},i} \frac{(1 - C_i)}{C_i} \sum_{\substack{l=n_k+1 \\ m \not \in \zeta_k}}^{N} P_{l,m} + \sum_{i=1}^{n_k} \sum_{j=1}^{i-1} \gamma_{\text{th},i} \gamma_{\text{th},j} \frac{(1 - C_j)}{C_j} \sum_{\substack{l=n_k+1 \\ m \not \in \zeta_k}}^{N} P_{l,m}  + \sum_{i=1}^{n_k} \sum_{j=1}^{i-1}  \sum_{l=1}^{j-1}\gamma_{\text{th},i} \gamma_{\text{th},j} \gamma_{\text{th},l} \frac{(1-\mathbf{C}_l)}{\mathbf{C}_l}|\mathbf{H}_{l,k} \sum_{\substack{l=n_k+1 \\ m \not \in \zeta_k}}^{N} P_{l,m} 
} \right)
}{1 - \sum_{k=1}^{M} \left( \frac {
\sum_{i=1}^{n_k} \gamma_{\text{th},i} \frac{(1 - C_i)}{C_i} \sum_{\substack{l=n_k+1 \\ m \not \in \zeta_k}}^{N} P_{l,m}  + \sum_{i=1}^{n_k} \sum_{j=1}^{i-1} \gamma_{\text{th},i} \gamma_{\text{th},j} \frac{(1 - C_j)}{C_j} \sum_{\substack{l=n_k+1 \\ m \not \in \zeta_k}}^{N} P_{l,m}  + \sum_{i=1}^{n_k} \sum_{j=1}^{i-1}  \sum_{l=1}^{j-1}\gamma_{\text{th},i} \gamma_{\text{th},j} \gamma_{\text{th},l} \frac{(1-\mathbf{C}_l)}{\mathbf{C}_l}|\mathbf{H}_{l,k} \sum_{\substack{l=n_k+1 \\ m \not \in \zeta_k}}^{N} P_{l,m} 
}{1 + 
\sum_{i=1}^{n_k} \gamma_{\text{th},i} \frac{(1 - C_i)}{C_i} \sum_{\substack{l=n_k+1 \\ m \not \in \zeta_k}}^{N} P_{l,m}  + \sum_{i=1}^{n_k} \sum_{j=1}^{i-1} \gamma_{\text{th},i} \gamma_{\text{th},j} \frac{(1 - C_j)}{C_j} \sum_{\substack{l=n_k+1 \\ m \not \in \zeta_k}}^{N} P_{l,m} + \sum_{i=1}^{n_k} \sum_{j=1}^{i-1}  \sum_{l=1}^{j-1}\gamma_{\text{th},i} \gamma_{\text{th},j} \gamma_{\text{th},l} \frac{(1-\mathbf{C}_l)}{\mathbf{C}_l}|\mathbf{H}_{l,k} \sum_{\substack{l=n_k+1 \\ m \not \in \zeta_k}}^{N} P_{l,m} 
} \right)}
\end{aligned}
$} 
\end{equation}
\end{figure*}

The overall power allocation strategy balances the power distribution within and between clusters, ensuring that total power remains within the maximum power limit \( P_{\text{Max}} \) while satisfying the QoS requirements for all users in the network.
We propose an iterative power control algorithm that dynamically adjusts \(P_{\text{max}}\) and refines the power allocation strategy over multiple iterations. The fine-tuned power allocation supports the QoS requirements and the successful execution of SIC, enabling the system to meet the demands of massive connectivity and high performance in wireless networks.
The proposed algorithm iteratively updates \(P_{\text{max}}\) based on the difference between \(P_{\text{max}}\) and \(P_{\text{min}}\) for each user \(i\), expressed as:

\begin{equation}
    \delta_i = P_{\text{max}_i} - P_{\text{min}_i}
\end{equation}

This disparity factor \(\delta_i\) serves as a feedback mechanism, known as the power allocation feedback (PAF), which ensures efficient resource management and dynamically adjusts users' transmit power. The power control algorithm, detailed in Algorithm \ref{alg:power_control}, uses an iterative optimization approach based on gradient descent to minimize the difference between the current and the desired power levels. This optimization enhances resource utilization and ensures adherence to system constraints.

\begin{algorithm}[!htbp]
\caption{Iterative Power Control}
\label{alg:power_control}
\begin{algorithmic}[1]
    \State Initialize $P_{max}$
    \State $P_{\text{min}} \gets \texttt{calcul\_P}$
    \State Compute the total power range $Pr \gets \sum P_{max} - \sum P_{min}$
    \For{each user \(i\)}
        \State Calculate the difference \(\delta_i \gets P_{\text{max}_i} - P_{\text{min}_i}\)
    \EndFor
    \For{each user \(i\)}
        \State Update \(P_{\text{max\_up}_i} \gets \left(\frac{Pr - \delta_i}{M-1}\right) + P_{\text{min}_i}\)
    \EndFor
    \State Replace \(P_{\text{max}}\) with \(P_{\text{max\_up}}\) for the next iteration
\end{algorithmic}
\end{algorithm}

\section{User Clustering}

\subsection{\textbf{User Pairing}}

\begin{figure}[!htbp]
    \centering 
    \includegraphics[width=\columnwidth]{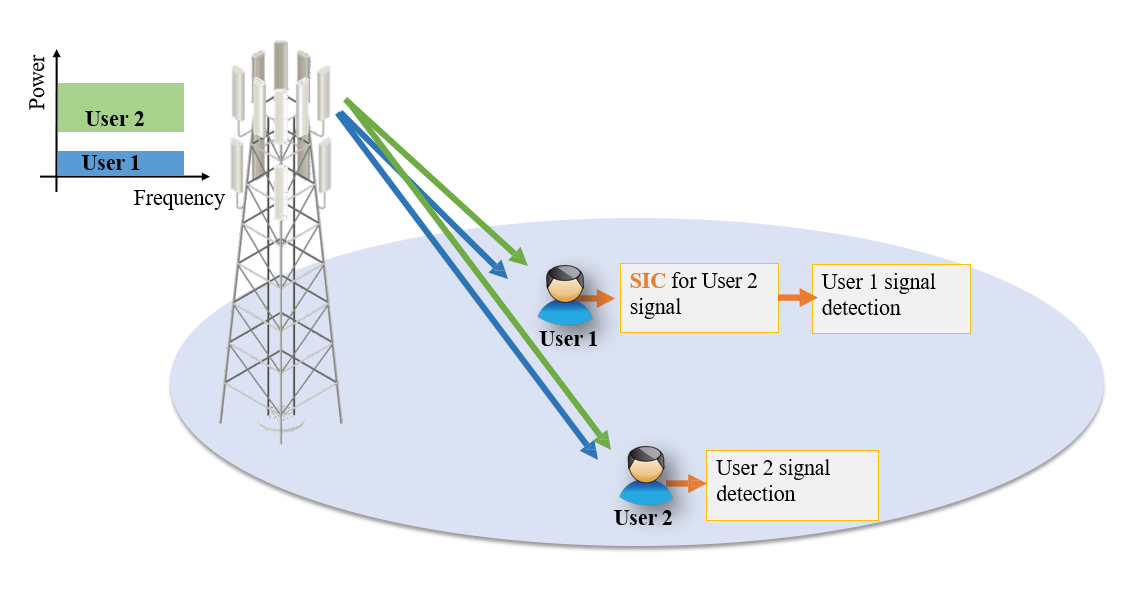}
    \caption{ 2-user PD\_{NOMA}}
   \label{fig noma}
\end{figure}
In MIMO-NOMA systems, user pairing selects which users to pair based on their channel conditions or spatial characteristics to optimize the utilization of the available spectrum. This section investigates two distinct user pairing methods based on channel gain and spatial correlation. Each method has unique characteristics and implications for the system's performance.
\subsubsection{Near-Far user pairing}
The first method pairs users based on their proximity to the base station; the user with the highest channel gain (or closest to the BS) is paired with the user who exhibits the lowest channel gain (or farthest from the BS). This pairing strategy exploits the channel gain disparities among users to maximize the system's spectral efficiency. Following the power control mechanism described above, power is allocated to ensure that users with lower channel gains receive a higher power allocation to compensate for their weaker signals. Conversely, users with stronger channel gain are allocated lower power levels, as demonstrated in Figure~\ref{fig noma}. This method promotes fairness among users and optimizes the utilization of available power resources within the cluster.

\subsubsection{Correlation-based user pairing }

The second method for user pairing is based on spatial correlation among users to significantly improve the overall system performance. This approach involves pairing users with high collinearity within the same cluster while ensuring low correlations with users in different clusters, thereby minimizing inter-cluster interference and enhancing the clarity and strength of transmitted signals, which is essential for SIC success.
By continuously adjusting power levels in response to QoS requirements and available power budget constraints, the power control algorithm ensures that each pair of users achieves the desired SINR. This alignment of power allocation and control strategies with the spatial characteristics of users maximizes the efficiency and fairness of the NOMA framework. It results in a more robust and reliable multi-user communication environment. As shown in Figure \ref{fig up}, Near-far pairing requires higher power than collinearity-based pairing due to limited pairing options for a sparse scenario. It tends to perform similarly for the dense scenario where more extensive pairing options are available.

\begin{figure}[!htbp]
    \centering \includegraphics[width=\columnwidth]{ 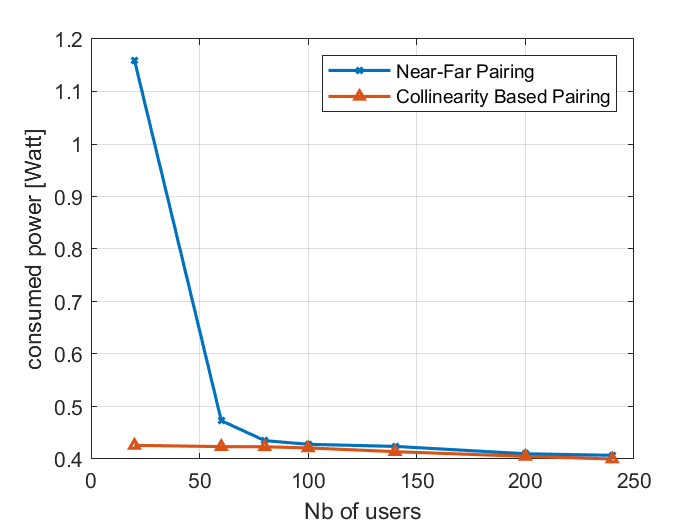}
    \caption{Comparison of user pairing approaches }
   \label{fig up}
\end{figure}

\subsection{\textbf{Correlation-based Iterative Algorithm (CIA)}}

The CIA is a clustering technique that groups users in a MIMO-NOMA system based on their spatial correlation, where correlated users who experience similar interference patterns are grouped into the same cluster. In contrast, users with low spatial correlation with more distinct and independent channel characteristics are placed in different clusters. Each cluster is served by a separate beam, using a linear zero-forcing beamforming technique, which focuses the beam toward the intended cluster while minimizing its impact on other clusters. 
The proposed CIA Algorithm comprises three phases: 
\begin{enumerate}
    \item \textbf{Initial User Group Formation:} \\
    The process begins by forming groups of users based on their spatial correlation. Users are paired according to their collinearity coefficients. The pair exhibiting the highest spatial correlation is placed in the first available cluster. The next pair is then assigned to the next available cluster until all clusters are occupied. This phase ensures that users with the strongest correlation factors are grouped together, leading to efficient beamforming and reduced inter-cluster interference.
    
    \item \textbf{Cluster Merging and Expansion:} In the second phase, the algorithm evaluates the possibility of combining current user groups to increase the number of users served within a single cluster. This phase examines the feasibility of merging two clusters while considering power constraints such as maximum power limits and interference restraints. After a successful merge, there may be an empty cluster where a new pair of correlated users or even a single user (singleton) can be added while still meeting power constraints. 
    
    \item \textbf{Final User Assignment and Power Control:} The final phase involves reviewing the list of remaining users who have not yet been assigned to any group. The algorithm tests, one by one, whether these users can be added to existing clusters without violating power constraints. Each potential addition is evaluated to ensure that the SIC process can still be performed successfully after the modification.
\end{enumerate}

The CIA dynamically performs these three phases through an iterative process to ensure that user clusters are optimally formed based on spatial correlation while simultaneously addressing power constraints and interference management. This mechanism significantly improves power usage, ensuring massive connectivity.
Dynamic power control further enhances the system's robustness by continuously adjusting power levels to meet evolving network conditions. This ultimately results in a highly efficient and reliable MIMO-NOMA system that can achieve spectral efficiency and fairness.

\begin{algorithm}[!htbp]
\caption{Correlation-based Iterative Algorithm (CIA)}
\label{alg:cia}
\begin{algorithmic}[1]
\State \textbf{Input:} Channel matrix $\mathbf{H}$, Maximum transmit power $P_{\text{Max}}$, Collinearity matrix $\mathbf{C}$, Sorted correlated pairs $\mathcal{C}$, Number of antennas $M$, Number of users $K$ and empty clusters $\zeta_i$
\State \textbf{Output:} Optimized clusters $\zeta$, Power allocation vector $P$, Number of served users $U$.

\State \textbf{Phase 1: Initial User Group Formation}
\For{each pair $(i,j)$ in $\mathcal{C}$}
    \For{$m \in \{1, \dots, M\}$}
        \If{$i$ is in $\zeta$ but not $j$}
            \State \text{Singletons} \(\gets \{j\}\)
        \ElsIf{$j$ is in $\zeta$ but not $i$}
            \State \texttt{Singletons} \(\gets \{i\}\)
        \ElsIf{neither $i$ nor $j$ are in $\zeta$}
            \State $P \gets$ \texttt{calcul\_P}($\zeta, \mathbf{H}, \mathbf{C}$)
            \If{\texttt{test\_P}($P$, $P_{\text{Max}}$) and $m \leq M$}
                \State $\zeta_m \gets \{i,j\}$
            \Else
                \State \text{Non\_alloc} \(\gets \{i,j\}\)
            \EndIf
        \EndIf
    \EndFor
\EndFor

\State \textbf{Phase 2: Cluster Merging and Expansion}
\For{$m \in \{1, \dots, M\}$}
    \For{$n \in \{1, \dots, M\}$}
        \For{$el \in$ non\_alloc}
            \State $\zeta_m \gets \zeta_m \cup \zeta_{n}$
            \State $\zeta_{n} \gets $ Non\_alloc($el$)
            \State \texttt{Calcul\_P ($\zeta, \mathbf{H}, \mathbf{C}$)}
            \If{\texttt{test\_P}($P$, $P_{\text{Max}}$)}
                \State keep this $\zeta$ structure
            \Else 
                \State Return to initial $\zeta$ structure
            \EndIf
        \EndFor
    \EndFor
\EndFor
\State \textbf{Phase 3: Final User Assignment and Power Control}
\For{$m \in \{1, \dots, M\}$}
    \For{$el \in $ singletons}
        \State $\zeta_m \gets$ singletons($el$)
        \State \texttt{Calcul\_P ($\zeta, \mathbf{H}, \mathbf{C}$)}
            \If{\texttt{test\_P}($P$, $P_{\text{Max}}$)}
                \State keep this $\zeta$ structure
            \Else 
                \State Return to initial $\zeta$ structure
            \EndIf
    \EndFor
\EndFor

\State \Return Final clusters $\zeta$, number of served users $U$, and power allocation $P$.
\end{algorithmic}
\end{algorithm}

The CIA algorithm is summarized in Algorithm~\ref{alg:cia}.  \texttt{Calcul\_P()}  computes the allocated power to each user in a cluster according to the power allocation equation provided in Equation~(\ref{eq: closed_form_pmin}). This function ensures that the power distribution among users is optimized to meet the desired communication requirements. \texttt{test\_P()} serves as a validation step, ensuring that the total power allocated to the cluster meets the power constraint of Equation~(\ref{eq: inegality}). 

\subsection{\textbf{K-means++ user clustering algorithm (KUC)}}

A robust approach to achieve the full potential of the massive MIMO and NOMA techniques is to perform the K-means++ algorithm to cluster users based on spatial correlation. K-means++ is particularly advantageous in our case due to its efficient initialization process, which helps the algorithm converge rapidly and provides better clustering results than standard K-means. The correlation factor reflects the degree of similarity in signal propagation characteristics and interference patterns. By grouping users with similar spatial correlation profiles, we are able to manage interference and improve spectral and energy efficiencies.

\begin{figure}[!htbp] 
\centering 
\includegraphics[width=\columnwidth]{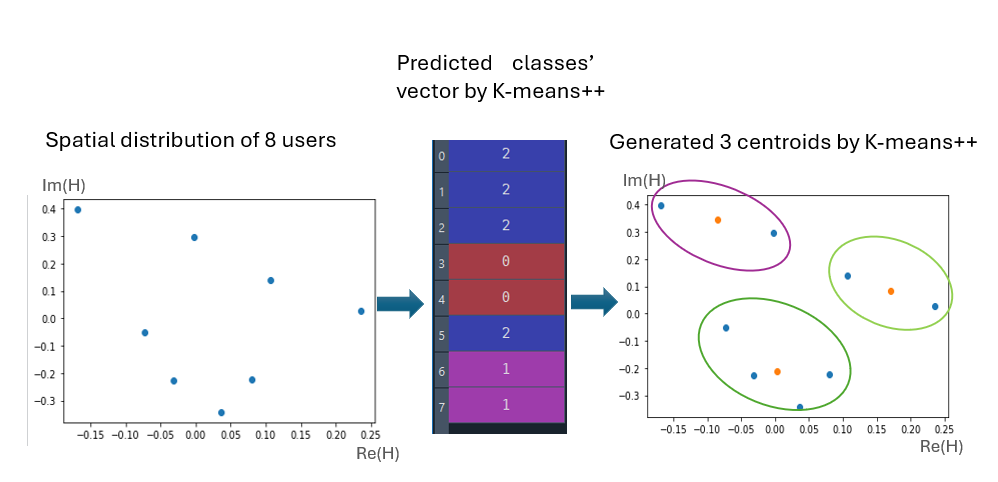} \caption{Example  of K-Means++ Clustering.} 
\label{fig:k-means++} 
\end{figure}

The proposed KUC is detailed in Algorithm \ref{alg:KUC}. We generate initial cluster centroids, representing the starting points for grouping users. These centroids are chosen using a probabilistic method that selects diverse starting points, leading to rapid convergence. These centroids denote the centers of each cluster in the multidimensional space defined by user attributes, such as channel gain and collinearity. The spatial correlation between users is calculated and used as input for the K-means++ algorithm. Each user is then assigned to the nearest centroid by calculating the Euclidean distance between the user's attribute values and the centroid's position, ensuring that users with similar characteristics are grouped together.
Once all users are assigned to the nearest centroid, the centroids are updated by recalculating their positions as the average values of all users currently assigned to that cluster. This iterative process continues until the algorithm converges.

This clustering approach is illustrated in Figure \ref{fig:k-means++}, where eight users are spatially distributed and then grouped into 3 clusters using K-means++. The generated centroids represent the optimal points around which users are clustered, ensuring that each group consists of users with highly correlated channels.

Following the clustering process, users within each cluster are allocated power while respecting the system's overall power budget. The power allocation algorithm ensures that the total power assigned to each cluster remains within the maximum allowable power limit, $P_{max}$, as defined by the power control process. 
The allocated power is then rigorously validated through the power control process to ensure that the distribution meets all constraints, such as maintaining the required SINR for each user and preventing excessive interference. This validation step is crucial for confirming that the allocated power levels satisfy QoS requirements and enable the SIC to be successful.

\begin{algorithm}[!htbp]
\caption{K-Means++ User Clustering (KUC)}
\label{alg:KUC}
\begin{algorithmic}[1]
\State \textbf{Input:} Channel matrix $\mathbf{H}$, Maximum transmit power $P_{\text{Max}}$, Collinearity matrix $\mathbf{C}$, Sorted correlated pairs $\mathcal{C}$, Number of antennas $M$, Number of users $K$ and empty clusters $\zeta_i$
\State \textbf{Output:} Optimized clusters $\zeta$, Power allocation vector $P$, Number of served users $U$.

\State \textbf{Phase 1: Initialize Centroids (K-means++ Initialization)}
\State  Select the first centroid $\mu_1$ randomly from the users.
\For{$i = 2$ to $M$}
    \State Select centroid $\mu_i$ with probability proportional to its squared distance from the nearest existing centroid:
    \State \( D_n \gets \min_{j < i} \sum_{p=1}^{M} \left(C[n, p] - C[\mu_j, p]\right)^2 \)
    \State \( \text{Probability}(n) \gets \frac{D_n}{\sum_{m=1}^{N} D_m} \)
\EndFor
\Repeat
\State \textbf{Phase 2: Assign Users to Centroids}
\For{each user $n \in \{1, \dots, N\}$}
    \For{each centroid $\mu_m$ where $m \in \{1, \dots, M\}$}
        \State Calculate the Euclidean distance $d_{n,m}$:
        \State \( d_{n,m} \gets \sqrt{\sum_{p=1}^{M} \left(C[n, p] - C[\mu_m, p]\right)^2} \)
    \EndFor
    \State Assign user $n$ to the cluster $\zeta_m$ corresponding to the minimum $d_{n,m}$.
\EndFor

\State \textbf{Phase 3: Update Centroids}
\For{each cluster $\zeta_m$ where $m \in \{1, \dots, M\}$}
    \State Update centroid $\mu_m$ as the mean of all users assigned to $\zeta_m$:
    \State \( \mu_m = \frac{1}{|\zeta_m|} \sum_{n \in \zeta_m} \mathbf{C}(n) \)
\EndFor
\Until{centroids stabilize (i.e., no significant change)}

\State \textbf{Phase 4: Refine Clusters Based on Power Constraints}
\State $\varphi$ represents the classes vector generated by K-Means++
\For{each class number $m \in \varphi $}
    \State $\zeta_m \gets \varphi(m)$.
    \State $P \gets$ \texttt{calcul\_P}($\zeta, \mathbf{H}, \mathbf{C}$)
    \If{\texttt{test\_P($P$, $P_{\text{max}}$)}}
        \State $\text{non-alloc} \gets  \{m\}$
        \State $\zeta_m \gets  \zeta_m  \setminus \{m\} $
    \EndIf
\EndFor

\State \textbf{Phase 5: Final Power Control}
    \State $P_{\text{max}} \gets$ \texttt{Control\_P($P$, $P_{\text{max}}$)}.

\State \Return Final clusters $\zeta$, number of served users $U$, and power allocation $P$.

\end{algorithmic}
\end{algorithm}

\subsection{\textbf{GWO-based User Clustering}}

The GWO is a robust metaheuristic algorithm inspired by the social hierarchy and hunting behavior of grey wolves in nature. Initially introduced by Mirjalili et al. \cite{Mirjalili2014}, GWO has gained popularity for its simplicity and efficiency in solving complex optimization problems in various domains. GWO has been proven to be highly effective in finding optimal solutions by balancing exploration (searching new areas of the solution space) and exploitation (refining the best-known solutions), which prevents the algorithm from getting trapped in local optima.
In our context, GWO is adapted for user clustering to maximize the number of served users while ensuring that the power allocation for each cluster respects the system's power budget constraints. The GWO algorithm operates by simulating the natural leadership hierarchy of grey wolves, where the best solutions are represented by alpha ($\alpha$), beta ($\beta$), and delta ($\delta$) wolves. The remaining wolves ($\omega$) follow these leaders by iteratively adjusting their positions accordingly. Each wolf's position represents a potential solution, i.e., vectors of user assignments where each vector corresponds to a specific clustering configuration. The effectiveness of the GWO in optimizing user clustering and power allocation tasks is mainly defined by its fitness function. This fitness function is carefully designed to maximize the number of served users while ensuring that the allocated power remains within the system's constraints. 
\begin{equation}
\label{eq:fitness}
\begin{split}
\text{Fitness} = \text{Served Users} - \left(\text{Non-Allocated Users} \times \text{Penalty}\right) \\
- \text{Penalty Power} \times \max\left(0, \text{Total Power} - M\right)
\end{split}
\end{equation}

Here,  \emph{Served Users} represents the number of users successfully clustered and served within the network. The second term, \emph{Non-Allocated Users} multiplied by a \emph{Penalty}, accounts for users who could not be allocated due to constraints. \emph{Penalty Power} multiplied by the excess power (\(\max(0, \text{Total Power} - M)\)) imposes an additional penalty if the total power allocated exceeds the maximum allowable power, \(M\). By integrating these components, the fitness function ensures that the GWO algorithm converges to configurations that maximize user allocation and maintain power efficiency, thus optimizing overall system performance.
\begin{figure}[!htbp] 
\centering 
\includegraphics[width=\columnwidth]{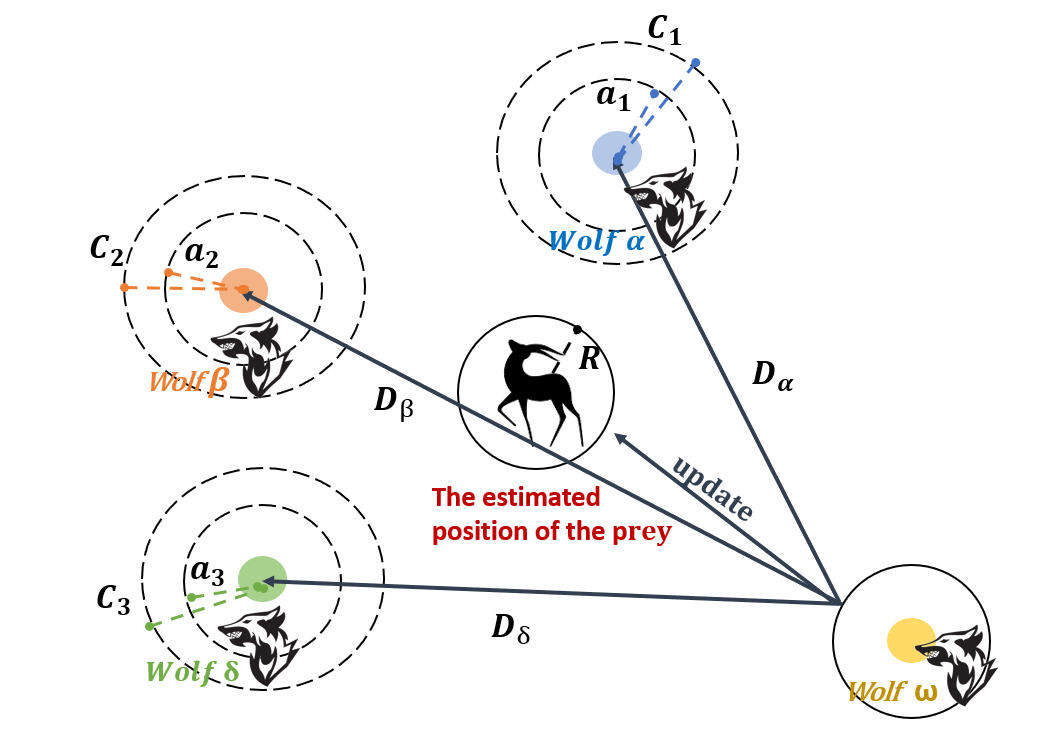} \caption{GWO hunting process.} 
\label{fig:gwo} 
\end{figure}
GWO was chosen for its effectiveness and rapid convergence to optimal solutions in complex optimization tasks. It operates through an iterative process that begins with the initialization of a population of wolves, where each wolf represents a potential solution $\mathbf{X} = [x_1, x_2, \dots, x_n]$. The initial positions of the wolves are randomly assigned, and the collinearity matrix $\mathcal{C}$ is calculated from the characteristics of the user channel to guide this initial clustering. The fitness of each wolf is then evaluated based on the fitness function in Equation (\ref{eq:fitness}) designed to maximize the number of served users while maintaining power constraints within the maximum allowable limit, $P_{\text{max}}$.
The wolves are ranked according to their fitness, with the top three solutions, $\alpha$, $\beta$, and $\delta$ wolves, guiding the remaining wolves.
The positions of all wolves are iteratively updated as they move towards $\alpha$, $\beta$, and $\delta$ wolves, simulating the process of hunting prey, where the prey represents the optimal solution $\mathbf{X}_{\text{prey}}$ as depicted in Figure \ref{fig:gwo}. This movement is mathematically formulated by calculating the distance $\mathbf{D}$ between the wolf's current position and the prey and then updating the wolf's position as follows:
\begin{equation}
    \mathbf{D} = |\mathbf{C} \cdot \mathbf{X}_{\text{prey}} - \mathbf{X}|,
\end{equation}
\begin{equation}
    \mathbf{X}_{\text{new}} = \mathbf{X}_{\text{prey}} - \mathbf{A} \cdot \mathbf{D},
\end{equation}
where $\mathbf{X}_{\text{new}}$ and $\mathbf{X}_{\text{prey}}$ denote the new and prey positions, respectively. The coefficient vectors $\mathbf{A}$ and $\mathbf{C}$ control the movement of the wolves, allowing for the balance of exploration and exploitation throughout the algorithm. These vectors are defined as:
\begin{equation}
    \mathbf{A} = 2\mathbf{a} \cdot \mathbf{r}_1 - \mathbf{a}, \quad \mathbf{C} = 2 \cdot \mathbf{r}_2,
\end{equation}

$\mathbf{a}$ linearly decreases from 2 to 0 throughout iterations, and $\mathbf{r}_1$ and $\mathbf{r}_2$ are random vectors in $[0, 1]$ that introduce variability into the wolves' movements, ensuring a proper balance between exploration (searching new areas of the solution space) and exploitation (refining potential or best-known solutions). After each position update, the allocated power to each cluster is recalculated, and the fitness of the new configuration is evaluated. If the total power allocated exceeds $P_{\text{max}}$, a penalty is applied to the fitness score, ensuring that the algorithm prioritizes solutions that adhere to power constraints. The algorithm checks for convergence by monitoring the changes in wolves' positions; if the positions stabilize, indicating that the wolves have converged on the optimal solution, the process stops. The final user clustering and power allocation configuration, represented by the alpha wolf position, is then selected as the optimal solution. This detailed iterative process optimizes user clustering and power allocation and improves system performance and efficient resource utilization in large-scale and complex MIMO-NOMA scenarios. GWO enables us to design effective and scalable solutions, leading to enhanced spectral and energy efficiency in networks with a high number of users.

\begin{algorithm}[!htbp]
\caption{Grey Wolf Optimizer for User Clustering (GWO)}
\label{alg:GWO}
\begin{algorithmic}[1]
\State \textbf{Input:} Channel matrix $\mathbf{H}$, Maximum transmit power $P_{\text{Max}}$, Collinearity matrix $\mathbf{C}$, Number of antennas $M$, Number of users $K$, Initial wolf pack assignments $\zeta_i$
\State \textbf{Output:} Optimized clusters $\zeta$, Power allocation vector $P$, Number of served users $U$.

\State \textbf{Phase 1: Initialization}
\State Initialize the grey wolf pack with random user-to-cluster assignments.
\State Set initial parameters $a$, $A$, and $C$ for exploration and exploitation.

\State \textbf{Phase 2: Main Loop}
\For{each iteration $i \in \{1, \dots, \text{max\_iterations}\}$}
    \State Update parameters $a_i$, $A_i$, and $C_i$:
    \State \( \mathbf{A_i} \gets 2\mathbf{a_i} \cdot \mathbf{r}_1 - \mathbf{a_i} \)
    \State \( \mathbf{C_i} \gets 2 \cdot \mathbf{r}_2 \)
    
    \State \textbf{Exploration Phase:}
    \For{each wolf $w \in \text{pack}$}
        \State Update position $X_w$ for exploration:
        \State \( X_w^{\text{explore}} \gets X_w - A_i \cdot |C_i \cdot X_{\text{prey}} - X_w| \)
        \State Generate preliminary clusters $\zeta_m$ based on $X_w^{\text{explore}}$.
    \EndFor

    \State \textbf{Exploitation Phase:}
    \For{each wolf $w \in \text{pack}$}
        \State Refine position $X_w$ for exploitation:
        \State \( X_w^{\text{exploit}} \gets X_w^{\text{explore}} - A_i \cdot |C_i \cdot X_{\text{best}} - X_w^{\text{explore}}| \)
        \State Generate final clusters $\zeta_m$ based on $X_w^{\text{exploit}}$.
        \State $P \gets$ \texttt{calcul\_P}($\zeta, \mathbf{H}, \mathbf{C}$)
        \State fitness $\gets$ \texttt{calcul\_fitness}($\zeta, P$)
        \If{\texttt{test\_P}($P$, $P_{\text{max}}$)}
            \State Update positions of $\alpha$, $\beta$, and $\delta$ wolves based on fitness.
        \EndIf
    \EndFor
\EndFor

\State \textbf{Phase 3: Final Output}
\State $P_{\text{final}} \gets$ \texttt{control\_P}($P$, $P_{\text{max}}$)
\State \Return Optimized clusters $\zeta$, Number of served users $U$, and Power allocation $P_{\text{final}}$.
\end{algorithmic}
\end{algorithm}

\texttt{calcul\_fitness()}  evaluates the number of successfully served users, imposes penalties for non-allocated users, and applies additional penalties for exceeding the maximum allowable power, as defined in Equation~(\ref{eq:fitness}). This evaluation guides the algorithm towards configurations that maximize user allocation efficiency and power management.

\section{Computational Complexity Analysis}
Optimizing downlink MIMO-NOMA systems using user clustering and power allocation techniques requires careful consideration of computational complexity. This section analyzes the CIA, KUC, and GWO-based user clustering methods.

\begin{table}[H]
\centering
\caption{Parameters for computation complexity calculation}
\begin{tabular}{|c|p{0.75\columnwidth}|}
\hline
\textbf{Parameter} & \textbf{Description} \\ \hline
$M$ & Number of clusters (also the number of antennas). \\ \hline
$K$ & Total number of users in the system. \\ \hline
$N$ & Number of extracted correlated pairs of users. \\ \hline
$U$ & Number of users per cluster. \\ \hline
$S$ & Size of the user assignment pack. \\ \hline
\end{tabular}
\vspace{0.5cm}
\label{tab:parameters}
\end{table}
The CIA algorithm comprises three phases: initial cluster formation, cluster merging and expansion, and final user assignment. In the initial phase, where user pairs are checked and assigned based on collinearity, CIA has a complexity of \(O(N \times M^2 \times U^2)\) due to the combined cost of assigning users to clusters and recalculating the power coefficients for each user in the cluster. The subsequent phase of merging and expanding clusters adds an additional complexity of \(O(M^4 \times U^2)\), mainly due to the iterative power calculations and adjustments needed to meet power constraints after each merge. Thus, the total complexity of the CIA algorithm can be expressed by \(O(M^4 \times U^2)\), which makes it particularly suitable for moderate-scale systems, as it efficiently minimizes interference within clusters but becomes computationally demanding as the number of clusters \(M\) grows.

GWO algorithm has a complexity of \(O( M^4 \times U^4)\), is driven by both the number of iterations and the operations required in each phase. In the exploration phase, where each wolf updates its position and clusters are preliminarily formed, the complexity is \(O(S \times M + K \times M)\) per iteration. The exploitation phase, which includes refined position updates and power allocation recalculations, has a complexity of \(O(S \times M + M^2 \times U^2)\), with the power allocation dominating due to recalculations for each user in a cluster. Since GWO iterates to convergence, the total complexity depends heavily on the number of iterations, with each iteration requiring a cost \(O(M^2 \times U^2)\) for power allocation adjustments across clusters. 

KUC aims to minimize intra-cluster interference by optimally assigning users to centroids, which has a complexity dominated by \(O(K \times M^2 + M^3 \times U^2)\). This complexity arises from five phases, beginning with centroid initialization, where selecting centroids based on distances to existing centroids requires \(O(K \times M)\) operations. The user assignments to the centroids and subsequent centroid updates contribute \(O(K \times M^2)\) since each user must be evaluated for the distance to each centroid, and the centroids are recalculated iteratively. However, the most computationally expensive part of KUC lies in the cluster refinement phase, which has a complexity of \(O(M^3 \times U^2)\) due to repeated power calculations and adjustments for users within clusters to meet power constraints. Thus, KUC has a lower complexity than CIA and GWO, which makes it suitable for large-scale antenna arrays equipped with hundreds of antennas.

\section{Performance Evaluation}
\label{Simulations}

This section provides a comprehensive performance evaluation of the proposed user clustering methods such as CIA, KUC, and GWO. As benchmark techniques, we consider Random clustering and User Pairing. The latter has been updated to consider user correlation in the user selection process. The former has also been updated to assign users to clusters randomly, but there is a SIC check after each allocation to avoid SIC failure. The benchmark techniques have been modified to ensure comparability with the proposed approaches. The review focuses on three performance indicators: the number of assigned users, the total transmit power, and the energy efficiency. These metrics are examined under the simulation parameters summarized in Table~\ref{table:parameters}.

\begin{table}[!ht]
\centering
\caption{Simulation parameters.}
\label{table:parameters}
\begin{tabular}{|l|l|}
\hline
\textbf{Parameter} & \textbf{Value} \\ \hline
SINR threshold \((\gamma_{th})\) & 10 dB \\ \hline
Noise spectral density \((\sigma)\) & \(10^{-2}\) W/Hz\\ \hline
Maximum per-antenna power \((P_{max})\) & 1 Watt \\ \hline
Bandwidth & 200 kHz \\ \hline
\end{tabular}
\end{table}

\subsection{ Number of Allocated Users}
\label{subsec:user_allocation}

\begin{figure}[!ht]
\centering
    \begin{overpic}[scale=0.65]{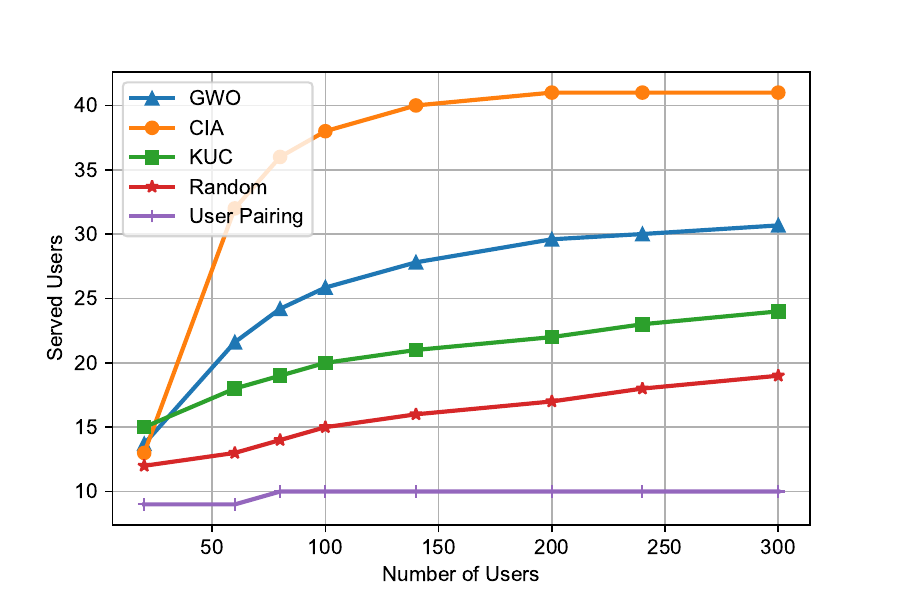}
    \end{overpic}
\caption{Number of served while varying the total number of users for $5$ RF chains. }
\label{fig:served_users_n_users}
\end{figure}

Figure \ref{fig:served_users_n_users} illustrates the number of served users varying the total number of users in the system with a number of RF chains fixed to $5$. CIA significantly outperforms the other algorithms, serving around $40$ users when the total number of active users exceeds $150$. GWO serves around $30$ users, approximately $25\%$ less than CIA, and KUC serves less than $25$ users, i.e., $37.5\%$ less than CIA. This performance enhancement is mainly due to the ability of the CIA to leverage spatial correlation among users, reducing inter-cluster interference in dense networks where user channels are more likely to be correlated. The random algorithm, as expected, serves at most $18$ users, $55\% $ less than the CIA, and User Pairing, restricted to serving a maximum of $2$ users per cluster, serves at most $10$ users, i.e., $75\% $ less than the CIA. Indeed, its fixed pairing capacity, even if it limits the interference in grant-free approaches, severely limits the scalability.

\begin{figure}[!ht]
\centering
    \begin{overpic}[scale=0.65]{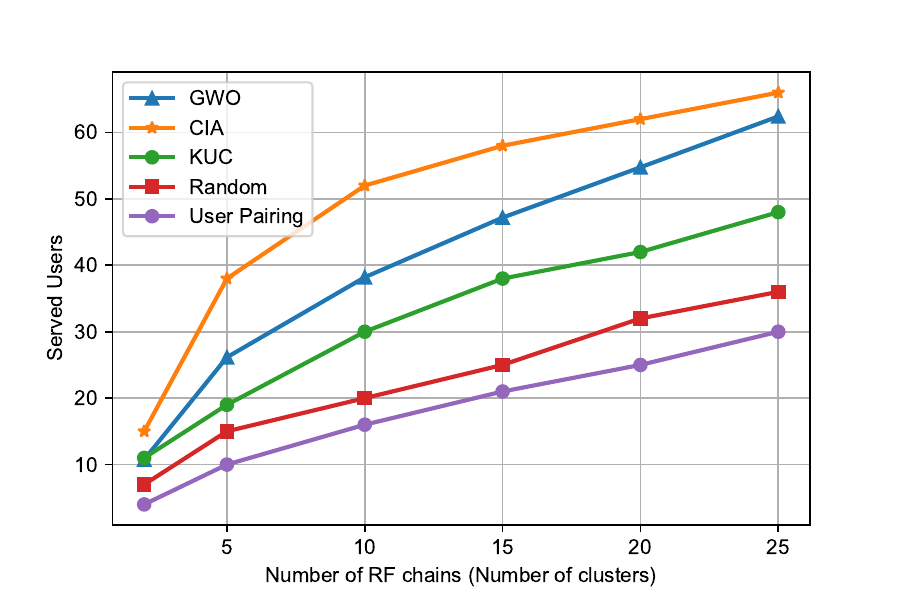}
    \end{overpic}
\caption{Number of served users while varying the number of RF chains for 100 users.}
\label{fig:served_users_n_antennas}
\end{figure}

In Figure \ref{fig:served_users_n_antennas}, we vary the number of RF chains (clusters) for a scenario where the total number of users is set to $100$. CIA, once again, demonstrates superior performance, serving about $67$ users with $25$ clusters. This trend highlights the effectiveness of CIA in leveraging spatial correlation and optimizing beamforming in massive MIMO systems to serve more users and maximize system capacity. GWO follows by serving $62$ users, which is $7.5\%$ less than CIA, benefiting from its adaptive nature and optimization strategies when increasing RF chains. 
KUC serves $50$ users,  $25.4\%$ less than CIA, and Random allocation serves up to $35$ users, i.e. $47.8\%$ less than CIA. User Pairing shows a linear increase in the number of served users with respect to the number of RF chains until inter-cluster and intra-cluster interference restrains filling all the clusters by user pairs. In fact, for $25$ RF chains, User Pairing serves only $30$ users, that is, $55.2\%$ less than the CIA.

\subsection{Transmit power}

\begin{figure}[!ht]
\centering
    \begin{overpic}[scale=0.65]{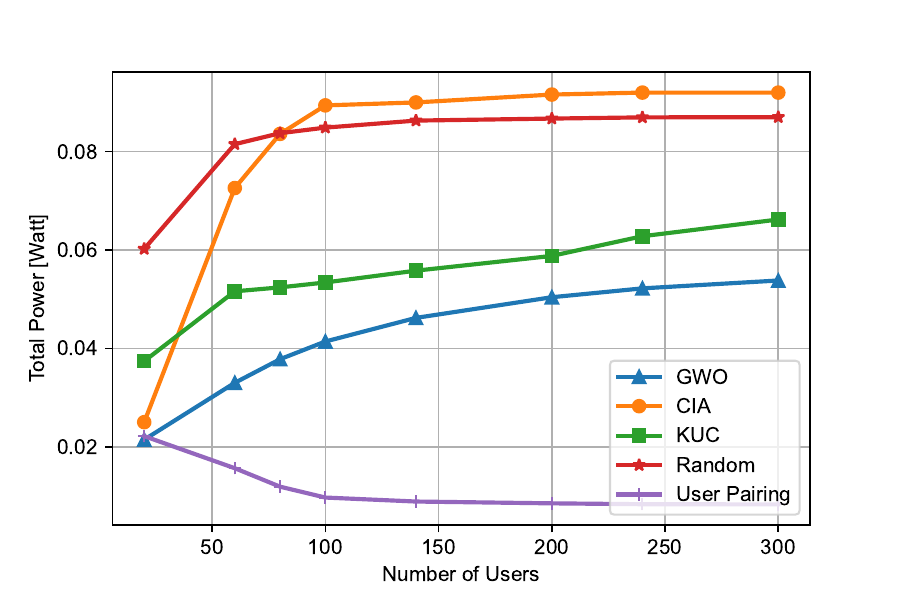}
    \end{overpic}
\caption{Total power consumption while varying the number of users for $5$ RF chains. }
\label{fig:served_power_n_users}
\end{figure}

Figure \ref{fig:served_power_n_users} illustrates the total power consumption depending on the number of users for $5$ RF chains. We can observe that the cost of efficiency of the CIA,  in terms of the number of served users, is at the expense of significant energy consumption. As we can see, for more than $100$ users, the total power consumption of the CIA is higher than all other techniques, around $0.09$ W.  Random allocation follows closely, consuming about $0.085$ W, and KUC shows moderate power consumption, around $0.065$ W for $300$ users ($0.054$ W for $100$ users), which is $22.2\%$ lower than CIA.
GWO, on the other hand, is particularly interesting as it serves a considerable number of users while consuming far less energy, around $0.05$ W between $100$ and $300$ users. This represents a $44.4\%$ decrease in power consumption compared to CIA, making GWO an attractive option for modern wireless networks where power efficiency and user service capacity are critical considerations. Finally, the User Pairing algorithm, which serves a maximum of $2$ users per cluster, maintains a steady power consumption of around $0.01$ W. Its lowest power consumption shows potential in energy-constrained scenarios. However, this may come at the cost of the number of served users. As the number of users increases, all algorithms reach their power constraints that limit further scalability.

\begin{figure}[!ht]
\centering
    \begin{overpic}[scale=0.65]{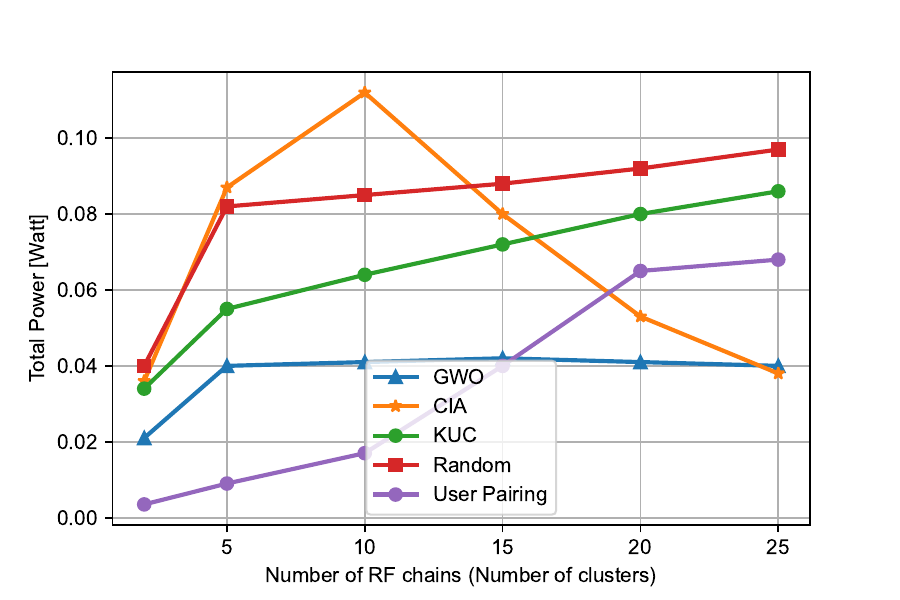}
    \end{overpic}
\caption{Total power consumption with varying RF chain numbers for 100 active users.}
\label{fig:served_power_n_antennas}
\end{figure}

Figure \ref{fig:served_power_n_antennas} shows the total power consumption as a function of the number of RF chains (clusters) for a scenario with $100$ users. CIA initially consumes significant power, reaching a peak of around $0.11$ W for the $10$ clusters. This high energy consumption is directly related to the growing number of served users, around $52$, as shown in Figure \ref{fig:served_users_n_antennas}. However, as the number of clusters increases beyond $10$, we reach the spatial correlation boundary, which means that further power optimization is no longer possible. This constraint arises from the large number of clusters and the absence of a substantial number of users with enough collinearity. Together, these factors result in a reduction in performance gains when more antennas are added.
Random allocation also consumes considerable power, reaching $0.09$ W on $10$ RF chains and maintaining this high consumption even as the number of clusters increases. 
KUC, which exhibits moderate power consumption, peaks at around $0.068$ W for $10$ clusters. This behavior reflects that KUC, while less greedy than CIA or Random, still requires significant power to maintain and optimize user clustering as the number of clusters increases.
GWO is the most power-efficient approach, consuming only about $0.04$ W at its peak. It minimizes power consumption by ensuring that only the necessary power is used to meet user demands, avoiding excess usage that is common in non-optimized settings. This makes GWO an attractive option for power-sensitive environments, offering a good balance between user service and power consumption.
It is trivial that User Pairing achieves lower energy consumption than the other techniques since it allocates at most $2$ users per cluster, which does not require significant power consumption to enable the signals to be distinguishable. In fact, the higher the NOMA reuse factor, the higher the power required since the first user to decode considers all the other signals' interference.

\subsection{Energy Efficiency}
Efficiency in energy usage is an essential metric for measuring system performance, particularly in the context of 5G and beyond, where user density constantly increases. It represents the total number of successfully transmitted data per unit of power consumed, typically measured in bits per joule (bits/J). The following expression will be used to measure energy efficiency in our research:
\begin{equation}
\label{eq:energy_efficiency}
\mathit{EE} = \frac{\sum\limits_{i=1}^{N} R_i}{P_{\text{total}}},
\end{equation}
where $R_i$ represents the achievable data rate for the $i^{th}$ user, which is a function of their allocated bandwidth $B_i$ and their experienced signal-to-interference-plus-noise ratio $\gamma_i$, and $P_{\text{total}}$ denotes the total power consumed by the system. High-throughput clustering techniques are compared for energy efficiency in wireless networks.

\begin{figure}[!ht]
\centering
    \begin{overpic}[scale=0.65]{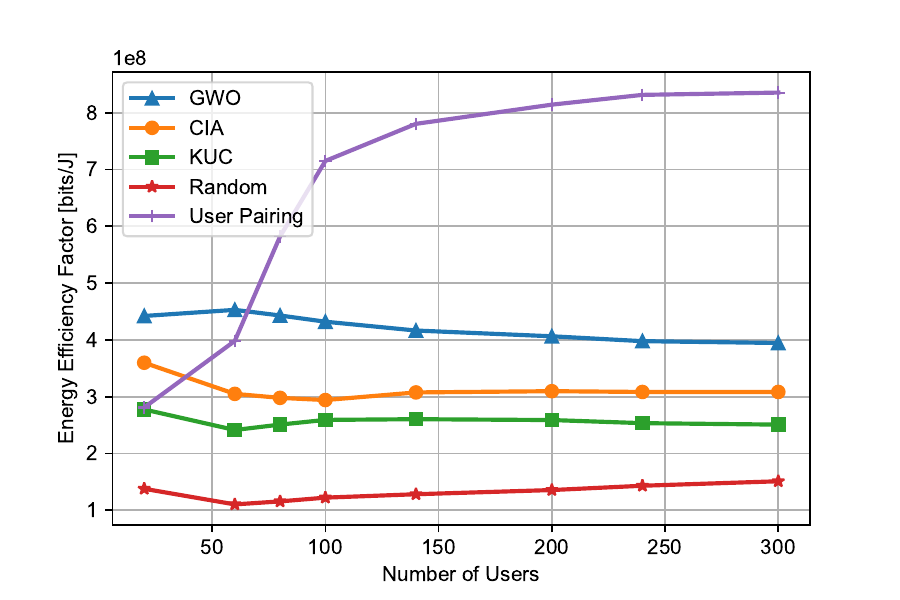}
    \end{overpic}
\caption{Energy efficiency while varying the number of users for a scenario of $5$ RF chains.}
\label{fig:energy_efficiency_users}
\end{figure}

Figure \ref{fig:energy_efficiency_users} illustrates the energy efficiency factor, measured in bits per joule, as the number of users increases while the number of RF chains remains fixed. User Pairing demonstrates the highest energy efficiency among algorithms and peaks at approximately $8 \times 10^8$ bits/J as the number of users increases. This behavior is somehow intuitive, as User Pairing serves only $2$ users per cluster and thus requires minimal power to distinguish signals while maintaining high energy efficiency.
GWO consistently shows strong energy efficiency, maintaining a value around $4 \times 10^8$ bits/J throughout the range of users. This reflects GWO's ability to balance power consumption and capacity. Its consistent performance demonstrates its ability to efficiently serve users without a significant increase in power consumption as the number of users grows.
CIA and KUC exhibit moderate energy efficiency, stabilizing around $3 \times 10^8$ bits/J and $2.5 \times 10^8$ bits/J, respectively, for $300$ active users. Although they group users effectively, their higher power requirements and clustering mechanism lead to a more modest energy efficiency than GWO and User Pairing.
Random allocation performs similarly, with energy efficiency factors near $1.5 \times 10^8$ bits/J, since it struggles with inefficient power usage due to its random nature.

\begin{figure}[!ht]
\centering
    \begin{overpic}[scale=0.65]{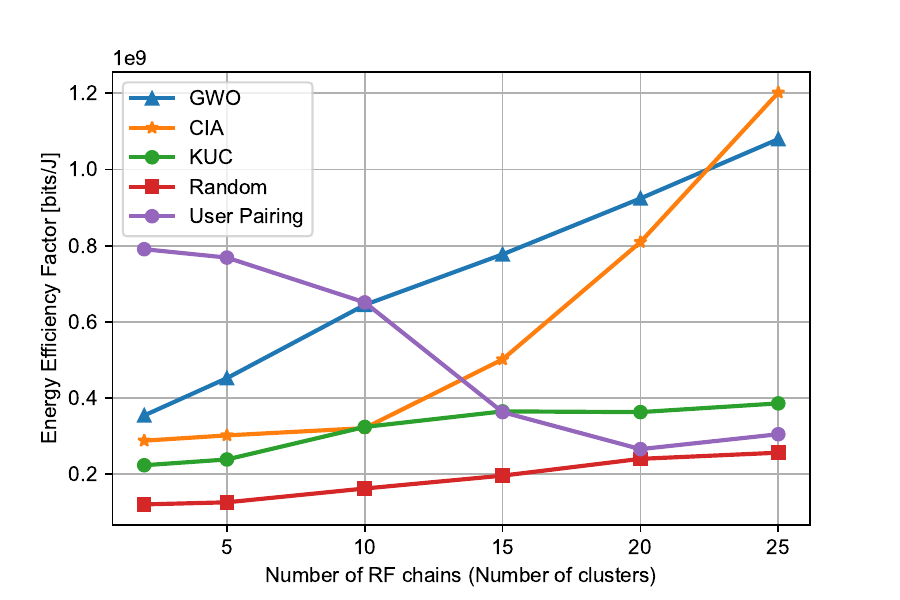}
    \end{overpic}
\caption{Energy efficiency while varying the number of RF chains for a scenario of $100$ active users.}
\label{fig:energy_efficiency_antennas}
\end{figure}

Figure \ref{fig:energy_efficiency_antennas} presents the energy efficiency factor as a function of the number of RF chains for $100$ users. The GWO algorithm demonstrates the highest energy efficiency, reaching approximately $1.1 \times 10^9$ bits/J when the number of clusters increases to $25$. This highlights its ability to maintain efficient power consumption as the number of clusters increases, making it a strong contender for scaling energy efficiency in larger systems.

CIA starts with low energy efficiency, but improves as the number of clusters increases, reaching a peak of around $1.2 \times 10^9$ bits/J when the number of RF chains reaches $25$. This is again due to its adaptivity to the spatial characteristics of the users, as illustrated in Figure~\ref{fig:served_power_n_antennas}.
The KUC maintains moderate energy efficiency, stabilizing at approximately $0.4 \times 10^9$ bits per joule as the number of clusters increases. This consistency indicates that KUC has a limited ability to optimize energy efficiency with additional clusters, reflecting a potential bottleneck in its scalability.
Random allocation displays relatively low energy efficiency, fluctuating around $0.15 \times 10^9$ bits/J and $0.25 \times 10^9$ bits/J, particularly in environments with increasing RF chain numbers. However, with few RF chains, the efficiency of User Pairing, which reaches $0.8 \times 10^9$ bits per joule, outperforms all proposed algorithms because it consumes the least amount of power, as shown in Figure \ref{fig:served_power_n_antennas}.


\section{Conclusion}
This study presented a detailed analysis of various user clustering and power allocation techniques within MIMO-NOMA systems, focusing on their applicability in advanced 5G and beyond networks. Integrating NOMA with massive MIMO technology is promising because it meets the stringent demands of modern wireless communications, particularly in scenarios requiring massive connectivity and high spectral efficiency.
The results of this research highlight the role of advanced clustering algorithms, such as GWO, CIA, and KUC, in optimizing the performance of MIMO-NOMA systems. In fact, GWO demonstrated exceptional adaptability and robustness in dynamic clustering scenarios, particularly in environments where spectral efficiency and user fairness are essential. On the other hand, the CIA was particularly effective in environments characterized by spatially correlated users. Meanwhile, the K-means++ algorithm, noted for its simplicity and computational efficiency, provides a practical solution for real-time applications, albeit with some limitations in environments with high user mobility and dynamic channel conditions.
An important insight from this research is the trade-off between computational complexity and system performance. Although sophisticated algorithms like GWO and CIA offer significant performance gains, particularly in terms of energy efficiency and system capacity, they also require greater computational resources. This complexity requires carefully considering the specific deployment scenario, balancing the need for advanced performance with the solution's computational feasibility.
Moreover, this study highlights the importance of integrating machine learning and metaheuristic optimization techniques in future 5G networks to balance system performance and computational efficiency. The research also points to the potential of these algorithms in addressing the challenges of massive connectivity and resource allocation in dense and dynamic network environments.

\bibliographystyle{elsarticle-num}
\bibliography{references}

\appendix

\section*{Appendix A: Power control}
\label{A1}

The SINR of user \( i \), denoted as \( \gamma_i \), is expressed as:

\begin{equation}
\left\{
\begin{array}{l}
    \gamma_i = \frac{C_{i} P_{i,k}|\mathbf{H}_{i,k} \mathbf{W}_{i,k}^{\text{H}}|^{2}}{\sigma^{2} + C_{i} |\mathbf{H}_{i,k} \mathbf{W}_{i,k}^{\text{H}}|^{2} P_{\text{intra}} + (1-C_{i})|\mathbf{H}_{i,k} \mathbf{W}_{i,k}^{\text{H}}|^{2} P_{\text{inter}}} \\
    \\
    P_{\text{intra}} = \sum_{j \in \zeta_k, j < i} P_{j,k}, \\
    \\
    P_{\text{inter}} = \sum_{\substack{l=i+1 \\ m \not \in \zeta_k}}^{N} P_{l,m},
\end{array}
\right.
\end{equation}

In this expression \( C_{i} \) represents the collinearity coefficient between the user \( i \) in the cluster \( k \) and the representative user of the cluster, \( p_{i,k} \) is the power coefficient allocated to the user \( i \) within the cluster \( k \), \( \mathbf{H}_{i,k} \) is the channel matrix and \( \mathbf{W}_{i,k} \) is the beamforming vector related to the user \( i \), while \( \sigma^{2} \) denotes the noise variance.

Let us consider the user \( n \),  the \( n \)-th user to decode its signal within the cluster \( \zeta_k \). The power allocated to the user \( n \), denoted by \( P_{n,k} \), must be carefully determined to ensure that their SINR \( \gamma_n \) meets or exceeds the predefined threshold \( \gamma_{\text{th}} \).
Rearranging the SINR inequality \( \gamma_n \geq \gamma_{\text{th}} \), we obtain the following expression for the required power allocation for the user \( n \):

\begin{equation}
\begin{aligned}
\label{eq:pn}
    P_{n,k} &\geq \frac{\gamma_{\text{th}} \sigma^{2}}{\mathbf{C}_{n} |\mathbf{H}_{n,k} \mathbf{W}_{n,k}^{\text{H}}|^{2}} + \gamma_{\text{th}} \sum_{j \in \zeta_k, j < n} P_{j,k} \\
    &\quad + \gamma_{\text{th}} \frac{(1-\mathbf{C}_{i})}{\mathbf{C}_{i}} \sum_{\substack{l=n+1 \\ m \not \in \zeta_k}}^{N} P_{l,m}.
\end{aligned}
\end{equation}

Here, the term \( \frac{\gamma_{\text{th}} \sigma^{2}}{\mathbf{C}_{n} |\mathbf{H}_{n,k} \mathbf{W}_{n,k}^{\text{H}}|^{2}} \) represents the minimum power required to overcome noise and ensure that the user's SINR exceeds the threshold \( \gamma_{\text{th}} \).
Meanwhile, the term \( \gamma_{\text{th}} \sum_{j \in \zeta_k, j < n} P_{j,k} \) accounts for the cumulative interference of users within the same cluster that is decoded before the user \( n \) and \( \gamma_{\text{th}} \frac{(1-\mathbf{C}_{n})}{\mathbf{C}_{n}} \sum_{l=n+1}^{N} P_{l,m} \) captures the interference from users outside the cluster.

The derived closed-form expression for \( P_{n,k} \) provides the required power allocation for each user while accounting for noise, intra-cluster, and inter-cluster interference. This expression is used in an iterative algorithm to allocate power dynamically and efficiently, ensuring that the total power constraint is respected while maximizing the number of users served.


\subsubsection{ Constant \( \gamma_{\text{th}} \) for all users}

Consider that the SINR threshold \( \gamma_{\text{th}} \) is the same for all users. The power allocation for the \( n \)-th user is derived in Equation (\ref{eq:pn}). By sequentially substituting for each \( P_{n,k} \) starting from the first user in the cluster, the system iteratively considers the cumulative effects of interference and adjusts the power allocations accordingly. This results in the following expressions for each user:
\begin{equation*}
\begin{aligned}
    P_{1,k} &\geq \frac{\gamma_{\text{th}} \sigma^{2}}{\mathbf{C}_{1} |\mathbf{H}_{1,k} \mathbf{W}_{1,k}^{\text{H}}|^{2}} + \gamma_{\text{th}} \frac {(1-\mathbf{C}_{1})}{\mathbf{C}_{1}} \sum_{\substack{l=2 \\ m \not \in \zeta_k}}^{N} P_{l,m} , \\
    P_{2,k} &\geq \frac{\gamma_{\text{th}} \sigma^{2}}{\mathbf{C}_{2} |\mathbf{H}_{2,k} \mathbf{W}_{2,k}^{\text{H}}|^{2}} + \gamma_{\text{th}} P_{1,k}+ \gamma_{\text{th}} \frac {(1-\mathbf{C}_{2})}{\mathbf{C}_{2}} \sum_{\substack{l=3 \\ m \not \in \zeta_k}}^{N} P_{l,m} , \\
    &\vdots \\
    P_{n,k} &\geq \frac{\gamma_{\text{th}} \sigma^{2}}{\mathbf{C}_{n} |\mathbf{H}_{n,k} \mathbf{W}_{n,k}^{\text{H}}|^{2}} + \gamma_{\text{th}} \sum_{i=1}^{n-1} P_{i,k} \\
    &\quad + \gamma_{\text{th}} \frac{(1-\mathbf{C}_{n})}{\mathbf{C}_{n}} \sum_{\substack{l=n+1 \\ m \not \in \zeta_k}}^{N} P_{l,m} .
\end{aligned}
\end{equation*}
Each \( P_{i,k} \) can be recursively expanded based on the allocated power to previous users. For \( P_{i,k} \), we substitute  Equation (\ref{eq:pn}) into the equation of \( P_{n,k} \):
\begin{equation*}
\begin{aligned}
    P_{n,k} &\geq \frac{\gamma_{\text{th}} \sigma^{2}}{\mathbf{C}_{n} |\mathbf{H}_{n,k} \mathbf{W}_{n,k}^{\text{H}}|^{2}} \\
    &\quad + \gamma_{\text{th}} \sum_{i=1}^{n-1} \left( \frac{\gamma_{\text{th}} \sigma^{2}}{\mathbf{C}_{i} |\mathbf{H}_{i,k} \mathbf{W}_{i,k}^{\text{H}}|^{2}} + \gamma_{\text{th}} \sum_{j=1}^{i-1} P_{j,k} \right. \\
    &\quad \left. + \gamma_{\text{th}} \frac{(1-\mathbf{C}_{i})}{\mathbf{C}_{i}} \sum_{\substack{l=n+1 \\ m \not \in \zeta_k}}^{N} P_{l,m} \right) \\
    &\quad + \gamma_{\text{th}} \frac{(1-\mathbf{C}_{n})}{\mathbf{C}_{n}} \sum_{\substack{l=n+1 \\ m \not \in \zeta_k}}^{N} P_{l,m}.
\end{aligned}
\end{equation*}

Then, we expand the summation and combine terms and obtain:

\begin{equation*}
\begin{aligned}
    P_{n,k} &\geq \frac{\gamma_{\text{th}} \sigma^{2}}{\mathbf{C}_{n} |\mathbf{H}_{n,k} \mathbf{W}_{n,k}^{\text{H}}|^{2}} + \gamma_{\text{th}}^2 \sum_{i=1}^{n-1} \frac{\sigma^{2}}{\mathbf{C}_{i} |\mathbf{H}_{i,k} \mathbf{W}_{i,k}^{\text{H}}|^{2}} \\
    &\quad + \gamma_{\text{th}}^2 \sum_{i=1}^{n-1} \sum_{j=1}^{i-1} P_{j,k} + \gamma_{\text{th}} \frac{(1-\mathbf{C}_{n})}{\mathbf{C}_{n}} \sum_{\substack{l=n+1 \\ m \not \in \zeta_k}}^{N} P_{l,m} \\
    &\quad + \gamma_{\text{th}}^2 \sum_{i=1}^{n-1} \frac{(1-\mathbf{C}_{i})}{\mathbf{C}_{i}} \sum_{\substack{l=n+1 \\ m \not \in \zeta_k}}^{N} P_{l,m}.
\end{aligned}
\end{equation*}

Notice that the double summation term \( \sum_{i=1}^{n-1} \sum_{j=1}^{i-1} P_{j,k} \) can be simplified since it represents a cumulative sum that involves powers of \( \gamma_{\text{th}} \). We can then simplify the expression as:
\begin{equation*}
\begin{aligned}
    P_{n,k} &\geq \frac{\gamma_{\text{th}} \sigma^{2}}{\mathbf{C}_{n} |\mathbf{H}_{n,k} \mathbf{W}_{n,k}^{\text{H}}|^{2}} + \gamma_{\text{th}}^2 \sigma^{2} \sum_{i=1}^{n-1} \frac{(1+\gamma_{\text{th}})^{n-i-1}}{\mathbf{C}_{i} |\mathbf{H}_{i,k} \mathbf{W}_{i,k}^{\text{H}}|^{2}} \\
    &\quad + \gamma_{\text{th}} \frac{(1-\mathbf{C}_{n})}{\mathbf{C}_{n}} \sum_{\substack{l=n+1 \\ m \not \in \zeta_k}}^{N} P_{l,m} \\
    &\quad + \gamma_{\text{th}}^2 \sum_{i=1}^{n-1} (1+\gamma_{\text{th}})^{n-i-1} \frac{(1-\mathbf{C}_{i})}{\mathbf{C}_{i}} \sum_{\substack{l=n+1 \\ m \not \in \zeta_k}}^{N} P_{l,m}.
\end{aligned}
\end{equation*}

This equation reflects the cumulative effect of previous power allocations, noise, and interference from other clusters, all modulated by the SINR threshold \( \gamma_{\text{th}} \) and the channel conditions represented by the matrices \( \mathbf{C}_{i} \) and \( \mathbf{H}_{i,k} \).
This iterative structure of the equations reveals how the power allocation for each user is influenced by the power levels allocated to previously decoded users and the interference from other clusters. It underscores the importance of a systematic optimization process to balance the power distribution and ensure that all users meet the quality of service requirements.

Finally, we impose a dual constraint to ensure that the total allocated power to all $n_k$ users within a cluster $k$, denoted as $A_{i} = \sum_{i=1}^{n_{k}} P_{i,k}$, does not exceed a maximum threshold $ P_{\text{Max}}$, taking into account the cumulative interference from other clusters represented by $A_{m \neq k_i}= \sum_{\substack{m=1 \\ m \neq k_i}}^{M} P_{i,m} $:

\begin{equation}
  \begin{aligned}
    & \left\{
    \begin{aligned}
      & A_{i} \geq \alpha_{i} + \beta_{i} A_{m \neq k_i}, \\
      & \sum_{k=1}^{M} A_{k} \leq P_{\text{Max}},
    \end{aligned}
    \right.
  \end{aligned}
\end{equation}

Where:

\begin{align}
\alpha_{i} &= \sum_{i=1}^{n_{k}} \frac{\gamma_{\text{th}} \sigma^{2}}{\mathbf{C}_{i} |\mathbf{H}_{i,k} \mathbf{W}_{i,k}^{\text{H}}|^{2}} \\
&\quad + \gamma_{\text{th}}^{2} \sigma^{2} \sum_{i=1}^{n_{i}} \sum_{j=1}^{i-1} \frac{(1+\gamma_{\text{th}})^{i-j-1}}{\mathbf{C}_{j} |\mathbf{H}_{j,k} \mathbf{W}_{j,k}^{\text{H}}|^{2}}, \\
\beta_{i} &= \sum_{i=1}^{n_{k}} \gamma_{\text{th}} \frac{(1-\mathbf{C}_{i})}{\mathbf{C}_{i}} \sum_{\substack{l=n+1 \\ m \not \in \zeta_k}}^{N} P_{l,m} \\
&\quad  + \gamma_{\text{th}}^{2} \sum_{i=1}^{n_{i}} \sum_{j=1}^{i-1} (1+\gamma_{\text{th}})^{i-j-1} \frac{(1-\mathbf{C}_{j})}{\mathbf{C}_{j}} \sum_{\substack{l=n+1 \\ m \not \in \zeta_k}}^{N} P_{l,m}.
\end{align}
where \( \alpha_{i} \) represents the minimum power requirements to satisfy the SINR threshold \( \gamma_{\text{th},i} \), considering the cumulative effect of noise and intra-cluster interference, while \( \beta_{i} \) captures the inter-cluster interference effects, accounting for the interference from other clusters and how the power allocation must adjust for users outside the cluster. These constraints ensure that power is optimally distributed within and across clusters while balancing satisfying individual QoS requirements and adhering to the overall power budget.

To simplify the power allocation constraints, we proceed with some mathematical simplifications in the first inequality, which represents the minimum power required to satisfy the cumulative needs of all clusters, by adding the term $\beta_i A_i $ on both sides. 
\begin{equation}
  \begin{aligned}
    & \left\{
    \begin{aligned}
      & (1+\beta_{i}) A_{i} \geq \alpha_{i}+ \beta_{i}  A_{m \neq K_i} + \beta{i} A_i, \\
      & \sum_{k=1}^{M} A_{k} \leq P_{\text{Max}},
    \end{aligned}
    \right.
    \end{aligned}
\end{equation}

The overall power allocation strategy integrates the constraints of all cluster $k \in {1,2,.., M}$ into an aggregated power constraint, which can be expressed as follows:

\begin{equation} 
\label{eq: system}
      \begin{aligned}
    & \left\{
    \begin{aligned}
      & \sum_{k=1}^{M} A_{k} \geq \frac{ \sum_{k=1}^{M} \frac{ \alpha_{k}}{(1+\beta_{k})}}{1-\sum_{k=1}^{M} \frac{\beta_{k}}{(1+\beta_{k})}}, \\
      & \sum_{k=1}^{M} A_{k} \leq P_{\text{Max}},
    \end{aligned}
    \right.
  \end{aligned}
\end{equation}

In this formulation, the first inequality represents the minimum power required to satisfy the cumulative requirements of all clusters, adjusted by the factors \( \alpha_{i} \) and \( \beta_{i} \) accounting for weaker users within the same cluster and the interference between clusters, respectively. The second inequality ensures that the total power allocated to all users does not exceed the network's maximum power limit \( P_{\text{Max}} \).

By synthesizing these constraints, the final equation encapsulates the dual objective of the optimization process: allocating sufficient power to meet or exceed the adjusted power requirements of all clusters while ensuring that the total network power remains within the allowed maximum.

\begin{equation}
\label{eq: inegality}
    P_{\text{Max}} \geq \sum_{k=1}^{M} A_{k} \geq P_{\text{min}},
\end{equation}
where \( P_{\text{min}} \) represents the minimum power required to meet the SINR thresholds and QoS requirements across all clusters.
When power is allocated, it is essential to find the balance between power distribution and network capacity to meet the QoS constraints of all users without overloading the network. This balanced approach allows the system to handle many connections and high data rates, meeting the needs of modern wireless networks.


\subsubsection{Different  \( \gamma_{\text{th}} \) for each user}
Considering different SINR thresholds \( \gamma_{\text{th}} \)  for each user introduces a more realistic setting, reflecting diverse QoS requirements across users. The power allocation expression for the \( n \)-th user is derived as follows:

\begin{equation}
\begin{aligned}
\label{pn_variable_gamma}
    P_{n,k} &\geq \frac{\gamma_{\text{th},n} \sigma^{2}}{\mathbf{C}_{n} |\mathbf{H}_{n,k} \mathbf{W}_{n,k}^{\text{H}}|^{2}} + \gamma_{\text{th},n} \sum_{i \in \zeta_{k}, i < n} P_{i,k} \\
    &\quad + \gamma_{\text{th},n} \frac{(1-\mathbf{C}_{n})}{\mathbf{C}_{n}} \sum_{\substack{l=n+1 \\ m \not \in \zeta_k}}^{N} P_{l,m}.
\end{aligned}
\end{equation}

By sequentially substituting for each \( P_{n,k} \) starting from the first user in the cluster, the system iteratively accounts for the cumulative effects of interference and adjusts the power allocations accordingly. This results in the following expressions for each user:

\begin{equation*}
\begin{aligned}
    P_{1,k} &\geq \frac{\gamma_{\text{th},1} \sigma^{2}}{\mathbf{C}_{1} |\mathbf{H}_{1,k} \mathbf{W}_{1,k}^{\text{H}}|^{2}} + \gamma_{\text{th},1} \frac {(1-\mathbf{C}_{1})}{\mathbf{C}_{1}} \sum_{\substack{l=2 \\ m \not \in \zeta_k}}^{N} P_{l,m}, \\
    P_{2,k} &\geq \frac{\gamma_{\text{th},2} \sigma^{2}}{\mathbf{C}_{2} |\mathbf{H}_{2,k} \mathbf{W}_{2,k}^{\text{H}}|^{2}} + \gamma_{\text{th},2} P_{1,k} + \gamma_{\text{th},2} \frac {(1-\mathbf{C}_{2})}{\mathbf{C}_{2}} \sum_{\substack{l=3 \\ m \not \in \zeta_k}}^{N} P_{l,m}, \\
    &\vdots \\
    P_{n,k} &\geq \frac{\gamma_{\text{th},n} \sigma^{2}}{\mathbf{C}_{n} |\mathbf{H}_{n,k} \mathbf{W}_{n,k}^{\text{H}}|^{2}} + \gamma_{\text{th},n} \sum_{i=1}^{n-1} P_{i,k} \\
    &\quad + \gamma_{\text{th},n} \frac{(1-\mathbf{C}_{n})}{\mathbf{C}_{n}} \sum_{\substack{l=n+1 \\ m \not \in \zeta_k}}^{N} P_{l,m}.
\end{aligned}
\end{equation*}

By extracting \( P_{i,k} \)  based on Equation (\ref{pn_variable_gamma})
and substituting it recursively based on the power assigned to previous users into the equation for \( P_{n,k} \), we obtain:

\begin{equation*}
\begin{aligned}
    P_{n,k} &\geq \frac{\gamma_{\text{th},n} \sigma^{2}}{\mathbf{C}_{n} |\mathbf{H}_{n,k} \mathbf{W}_{n,k}^{\text{H}}|^{2}} \\
    &\quad + \gamma_{\text{th},n} \sum_{i=1}^{n-1} \left( \frac{\gamma_{\text{th},i} \sigma^{2}}{\mathbf{C}_{i} |\mathbf{H}_{i,k} \mathbf{W}_{i,k}^{\text{H}}|^{2}} + \gamma_{\text{th},i} \sum_{j=1}^{i-1} P_{j,k} \right. \\
    &\quad \left. + \gamma_{\text{th},i} \frac{(1-\mathbf{C}_{i})}{\mathbf{C}_{i}} \sum_{\substack{l=n+1\\ m \not \in \zeta_k}}^{N} P_{l,m} \right) \\
    &\quad + \gamma_{\text{th},n} \frac{(1-\mathbf{C}_{n})}{\mathbf{C}_{n}} \sum_{\substack{l=n+1\\ m \not \in \zeta_k}}^{N} P_{l,m}.
\end{aligned}
\end{equation*}
and by combining the terms, we obtain:
\begin{equation*}
\begin{aligned}
    P_{n,k} &\geq \frac{\gamma_{\text{th},n} \sigma^2}{\mathbf{C}_n |\mathbf{H}_{n,k} \mathbf{W}_{n,k}^{\text{H}}|^2} \\
    &\quad + \sum_{i=1}^{n-1} \gamma_{\text{th},n} \gamma_{\text{th},i} \frac{\sigma^2}{\mathbf{C}_i |\mathbf{H}_{i,k} \mathbf{W}_{i,k}^{\text{H}}|^2} \\
    &\quad + \sum_{i=1}^{n-1} \sum_{j=1}^{i-1} \gamma_{\text{th},n} \gamma_{\text{th},i} \gamma_{\text{th},j} \frac{\sigma^2}{\mathbf{C}_j |\mathbf{H}_{j,k} \mathbf{W}_{j,k}^{\text{H}}|^2} \\
    &\quad + \gamma_{\text{th},n} \frac{(1 - \mathbf{C}_n)}{\mathbf{C}_n} \sum_{\substack{l=n+1\\ m \not \in \zeta_k}}^{N} P_{l,m} \\
    &\quad + \sum_{i=1}^{n-1} \gamma_{\text{th},n} \gamma_{\text{th},i} \frac{(1 - \mathbf{C}_i)}{\mathbf{C}_i} \sum_{\substack{l=n+1\\ m \not \in \zeta_k}}^{N} P_{l,m} \\
    &\quad + \sum_{i=1}^{n-1} \sum_{j=1}^{i-1} \gamma_{\text{th},n} \gamma_{\text{th},i} \gamma_{\text{th},j} \frac{(1 - \mathbf{C}_j)}{\mathbf{C}_j} \sum_{\substack{l=n+1\\ m \not \in \zeta_k}}^{N} P_{l,m}.
\end{aligned}
\end{equation*}

This formulation reflects the cumulative effect of previous power allocations, noise, and interference from other clusters, all modulated by user-specific SINR thresholds \( \gamma_{\text{th},n} \) and the channel conditions represented by the matrices \( \mathbf{C}_n \) and \( \mathbf{H}_{n,k} \).

The iterative structure of these equations shows how the power allocation for each user is influenced by the power levels allocated to previously decoded users and the interference from other clusters. It highlights the importance of a systematic optimization process to balance the power distribution and ensure that all users meet their respective quality of service requirements.

\begin{align}
\alpha_i &= \sum_{i=1}^{n_k} \left( \frac{\gamma_{\text{th},i} \sigma^2}{C_i |\mathbf{H}_{i,k} \mathbf{W}_{i,k}^{\text{H}}|^2} \right) 
+ \sum_{i=1}^{n_k} \sum_{j=1}^{i-1} \gamma_{\text{th},i} \gamma_{\text{th},j} \frac{\sigma^2}{C_j |\mathbf{H}_{j,k} \mathbf{W}_{j,k}^{\text{H}}|^2} \\
&\quad + \sum_{i=1}^{n_k} \sum_{j=1}^{i-1} \sum_{l=1}^{j-1} \gamma_{\text{th},i} \gamma_{\text{th},j} \gamma_{\text{th},l} \frac{\sigma^2}{\mathbf{C}_l |\mathbf{H}_{l,k} \mathbf{W}_{l,k}^{\text{H}}|^2} \\
\beta_i &= \sum_{i=1}^{n_k} \gamma_{\text{th},i} \frac{(1 - C_i)}{C_i} \sum_{\substack{l=n+1\\ m \not \in \zeta_k}}^{N} P_{l,m} \\
&\quad + \sum_{i=1}^{n_k} \sum_{j=1}^{i-1} \gamma_{\text{th},i} \gamma_{\text{th},j} \frac{(1 - C_j)}{C_j} \sum_{\substack{l=n+1\\ m \not \in \zeta_k}}^{N} P_{l,m} \\
&\quad + \sum_{i=1}^{n_k} \sum_{j=1}^{i-1}  \sum_{l=1}^{j-1}\gamma_{\text{th},i} \gamma_{\text{th},j} \gamma_{\text{th},l} \frac{(1-\mathbf{C}_l)}{\mathbf{C}_l}|\mathbf{H}_{l,k} \sum_{\substack{l=n+1\\ m \not \in \zeta_k}}^{N} P_{l,m}.
\end{align}

The closed-form expression for \( P_{\text{min}} \) is then expressed as follows:
\begin{equation}
P_{\text{min}} = \frac{\sum_{k=1}^{M} \left( \frac{\alpha_k}{1 + \beta_k} \right)}{1 - \sum_{k=1}^{M} \left( \frac{\beta_k}{1 + \beta_k} \right)}.
\end{equation}

\end{document}